\documentclass[pdflatex,sn-mathphys-num]{sn-jnl}% Math and Physical Sciences Numbered Reference Style
\usepackage{graphicx}%
\usepackage{multirow}%
\usepackage{amsmath,amssymb,amsfonts}%
\usepackage{amsthm}%
\usepackage{mathrsfs}%
\usepackage[title]{appendix}%
\usepackage{xcolor}%
\usepackage{textcomp}%
\usepackage{manyfoot}%
\usepackage{booktabs}%
\usepackage{algorithm}%
\usepackage{algorithmicx}%
\usepackage{algpseudocode}%
\usepackage{listings}%
\usepackage[utf8]{inputenc}

\theoremstyle{thmstyleone}%
\theoremstyle{thmstyletwo}%

\theoremstyle{thmstylethree}%

\begin{document}

\title[Article Title]{Hierarchical cell identities emerge from animal gene regulatory mechanisms}

%%=============================================================%%
%% GivenName	-> \fnm{Joergen W.}
%% Particle	-> \spfx{van der} -> surname prefix
%% FamilyName	-> \sur{Ploeg}
%% Suffix	-> \sfx{IV}
%% \author*[1,2]{\fnm{Joergen W.} \spfx{van der} \sur{Ploeg} 
%%  \sfx{IV}}\email{iauthor@gmail.com}
%%=============================================================%%

\author[1]{\fnm{Anton} \sur{Grishechkin}}

\author[1]{\fnm{Abhirup} \sur{Mukherjee}}

\author*[1]{\fnm{Omer} \sur{Karin}}\email{o.karin@imperial.ac.uk}

\affil*[1]{\orgdiv{Department of Mathematics}, \orgname{Imperial College London}, \orgaddress{\street{Exhibition Road, South Kensington}, \city{London}, \postcode{SW7 2AZ},  \country{UK}}}

%%==================================%%
%% Sample for unstructured abstract %%
%%==================================%%

\abstract{The hierarchical organisation of cell identity is a fundamental feature of animal development with rich and well-characterized experimental phenomenology, yet the mechanisms driving its emergence remain unknown. The regulation of cell identity genes relies on a distinct mechanism involving higher-order interactions of transcription factors on distant regulatory regions called enhancers. These interactions are mediated by epigenetic regulators that are broadly shared between enhancers. Through the development of a new and predictive mathematical theory on the effects of epigenetic regulator activity on gene network dynamics, we demonstrate that hierarchical identities are essential emergent properties of animal-specific gene regulatory mechanisms. Hierarchical identities arise from the interplay between enhancer competition for epigenetic readers and cooperation through activation of shared transcriptional programs.
We show that epigenetic regulatory mechanisms provide the network with self-similar properties that enable multilineage priming and signal-dependent control of progenitor states. The stabilisation of progenitor states is predicted to be controlled by the balance in activities between epigenetic writers and erasers. Our model quantitatively predicts lineage relationships, reconstructs all known blood progenitor states from terminal states, and explains mechanisms of cell identity dysregulation in cancer and the general differentiation effects of histone deacetylase inhibition. We identify non-specific modulation of enhancer competition as a central regulatory axis, with implications for developmental biology, cancer, and differentiation therapy.
}

\maketitle

\section*{Introduction}

Animals are composed of a wide range of cell types, each with specialised functions that collectively enable the formation of complex systems such as the brain and immune system. During development, cell types emerge through processes of hierarchical differentiation, in which precursor cells progressively adopt distinct identities \cite{waddington1957strategy,briggs2018dynamics}. These processes continue in many adult tissues, including the blood, where stem cells give rise to diverse lineages \cite{maansson2007molecular,ceredig2009models,dzierzak2018blood,weinreb2020lineage}, with cells in these lineages further specialising in response to signals \cite{allman2008peripheral,bluestone2009functional,guilliams2018developmental,ballesteros2020co}. Similarly, neural development generates many distinct neuron types through sequential specification \cite{blankvoort2018marked,soldatov2019spatiotemporal,li2020transcriptional,sharma2020emergence,faure2020single,laddach2023branching,li2024spatiotemporal}.

The hierarchical nature of cell identity is reflected in the underlying dynamics of gene expression and chromatin state. Recent studies revealed that individual cell types exhibit distinct high-dimensional gene expression programs, which are associated with cell-type specific activation of genomic regulatory elements called enhancers \cite{creyghton2010histone,whyte2013master,hnisz2013super,heinz2015selection,saint2016models,blankvoort2018marked}. Along differentiation hierarchies, progenitor cells often co-express at lower levels the expression programs of their target progeny, a phenomenon known as multilineage priming \cite{hu1997multilineage, miyamoto2002myeloid, mercer2011multilineage,brunskill2014single,ohnishi2014cell, kim2014broadly,olsson2016single, zheng2018molecular,briggs2018dynamics,soldatov2019spatiotemporal,martin2021chromatin,singh2022cell} (Figure ~\ref{fig:hierarchical}A). Priming is tightly linked to differentiation potential, with signals that bias toward specific fates modulating the contribution of fate-specific programs in progenitor states. For example, erythropoietin enhances the erythrocyte program in erythrocyte-primed progenitors, while neural crest progenitor potential correlates with relative target program expression \cite{grover2014erythropoietin,herman2018fateid,soldatov2019spatiotemporal,erickson2023transcriptional} (Figure ~\ref{fig:hierarchical}B).

While hierarchical differentiation is often depicted as a tree-like structure, akin to evolutionary phylogeny, this view does not align with our current understanding of cell type specification. For many cell types, fate convergence occurs, where a single cell type can arise through multiple differentiation trajectories \cite{weinreb2020lineage,erickson2023transcriptional,erickson2024unbiased} (Figure ~\ref{fig:hierarchical}C). This is evident in the blood, where many progenitor cell types with overlapping fate potentials have been identified \cite{ceredig2009models}. During regeneration, differentiated cells can undergo de-differentiation that can recapitulate hierarchical differentiation programs \cite{singh2022cell}. Finally, mutations can lead to cellular transformation, sometimes stabilising novel progenitor states, including hybrid identities that blend stem and differentiated expression programs \cite{diffner2013activity,smith2018human} (Figure ~\ref{fig:hierarchical}D).

The ubiquity of hierarchical differentiation raises the question of how it emerges from the underlying molecular regulatory networks. Early models of gene regulatory networks argued that hierarchy may be encoded within the structure of the network itself, which was assumed to be both highly modular and hierarchical  \cite{bolouri2002modeling,betancur2010assembling,peter2011evolution}. However, in many systems, including the blood, gene interaction networks are densely interconnected, without clear hierarchy or modularity \cite{novershtern2011densely,hamey2017reconstructing,trevers2023gene}, and, across cell types, specification is combinatorial, with transcription factors frequently overlapping in expression \cite{hnisz2013super,arendt2016origin,reilly2020unique}. In addition, the underlying molecular mechanisms of animal gene regulation, in which transcription factors have limited specificity compared with their bacterial counterparts, appear to deviate substantially from the strict circuit image proposed in gene regulatory models \cite{newman2020cell}. Despite this complexity, hierarchical differentiation is a robust phenomena with distinct dynamical features. A particularly crucial question concerns the origin and nature of progenitor states - both in normal development and disease conditions.

Here, through the development of a predictive mathematical theory, we propose that hierarchical cell identity is an emergent property of the distinct mechanisms of animal gene regulation. Specifically, we show that hierarchical identity emerges naturally from the interplay between enhancer competition for epigenetic readers and enhancer cooperation in activating overlapping transcriptional programs, with the relative contribution of each process tuned by cells to facilitate hierarchical transitions. 

\begin{figure}[H]%[tbhp]
\centering
\includegraphics[width=0.68\linewidth]{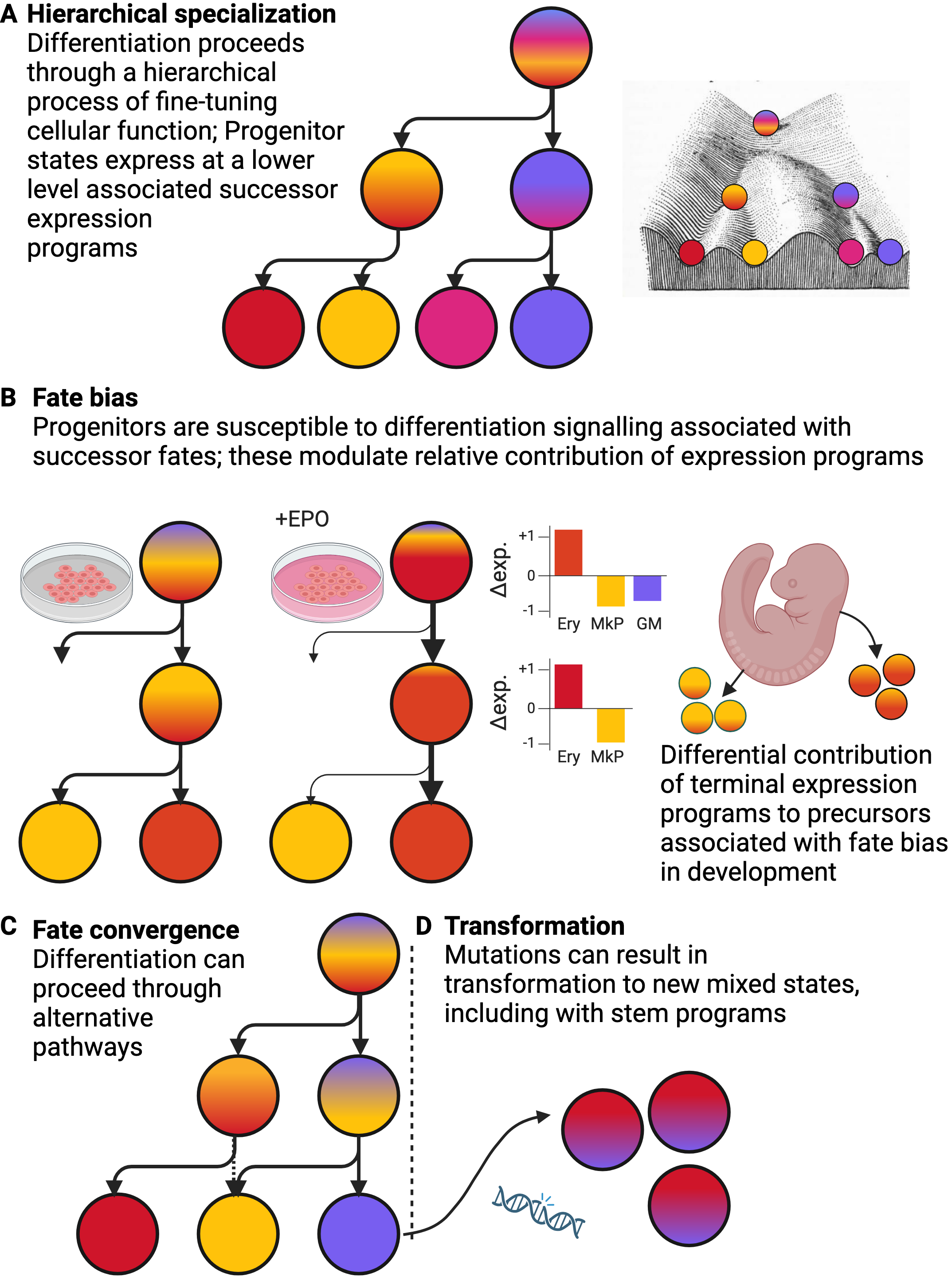}
\caption{\textbf{Features of hierarchical cell identity.} (A) Hierarchical differentiation, based on specific transitions between primed progenitors, is a hallmark of animal development and homeostasis. (B) Progenitors may display fate bias, where an increase in the likelihood of obtaining a target fate is associated with an increase in the relative contribution of the expression program of that fate to the progenitor state. (C) Some target fates may be achieved through alternative differentiation trajectories associated with progenitors with overlapping fate potential. (D) Mutation can result in transformation in which progenitor states are stabilised, and new overlapping progenitor states may be formed. }
\label{fig:hierarchical}
\end{figure}

\subsection*{Model for enhancer-regulated gene expression}

Animal cell types are associated with distinct gene expression patterns, which emerge from a regulatory mechanism based on the differential activation of enhancer elements \cite{creyghton2010histone,whyte2013master,hnisz2013super,narita2021enhancers,karin2024enhancernet}. Enhancers facilitate transcription by recruiting molecular machinery that initiates and enables transcription at distant genetic elements. A given genome may encode for hundreds of thousands of potential enhancer elements, yet, in a given cell type, only a small subset of these will be active and facilitate transcription \cite{heinz2015selection}. 

At the microscopic (single enhancer) level, the activation of enhancer elements involves a sequence of regulatory events unique to metazoans. The process begins with the recruitment of coactivator molecules, including the histone acetyltransferases (HATs) CBP/p300, to specific enhancer sites, by transcription factors \cite{weinert2018time,narita2021enhancers,ferrie2024p300} (Figure~\ref{fig:mechanism}A). These HATs function as epigenetic writers, acetylating histone proteins associated with enhancers. Acetylation is subsequently turned over by epigenetic erasers - histone deacetylases (HDACs), a process that occurs on a timescale of minutes \cite{weinert2018time,narita2021enhancers}. Histone acetylation then enables the recruitment of epigenetic reader molecules, which play a critical role in facilitating transcription. For genes associated with cell-type specific expression, the bromodomain protein Brd4 emerges as a key regulator, essential for recruiting factors that promote transcriptional elongation—a major rate-limiting step in gene expression \cite{chapuy2013discovery,loven2013selective,kanno2014brd4,heinz2015selection,winter2017bet,altendorfer2022brd4}, with other epigenetic readers, including Taf1, also implicated in transcriptional control \cite{narita2021enhancers}. (Figure~\ref{fig:mechanism}B).

At the macroscopic (many enhancer) level, enhancer activities are interconnected through the shared use of transcription factors and cofactor molecules. Although animal transcription factors can recognize and bind specific DNA sequences, their interactions are transient, targeting short and degenerate motifs \cite{ferrie2022structure}. As a result, a single transcription factor may bind to thousands of enhancer elements across the genome. In contrast, transcriptional cofactors do not recognize specific DNA sequences, and their activity is broadly distributed across multiple enhancers. This coupling of enhancers through cofactors occurs at both the epigenetic writing and reading stages. While HDACs are highly abundant, coactivators such as CBP/p300 are often stoichiometrically limiting compared to transcription factors, leading competition between transcription factors over their recruitment \cite{gillespie2020absolute,polansky2020dynamics}. In addition, bromodomain proteins exhibit high mobility, and their rapid redistribution in response to differential enhancer acetylation plays a critical role in modulating enhancer activity \cite{dey2003double,hah2011rapid,step2014anti,brown2014nf,schmidt2015acute,lee2017brd4,mishra2017histone,marie2018hdac,brown2018bet,slaughter2021hdac} (Figure~\ref{fig:mechanism}C). Together, these mechanisms determine gene expression at the whole-cell level.

\begin{figure}%[tbhp]
\centering
\includegraphics[width=0.75\linewidth]{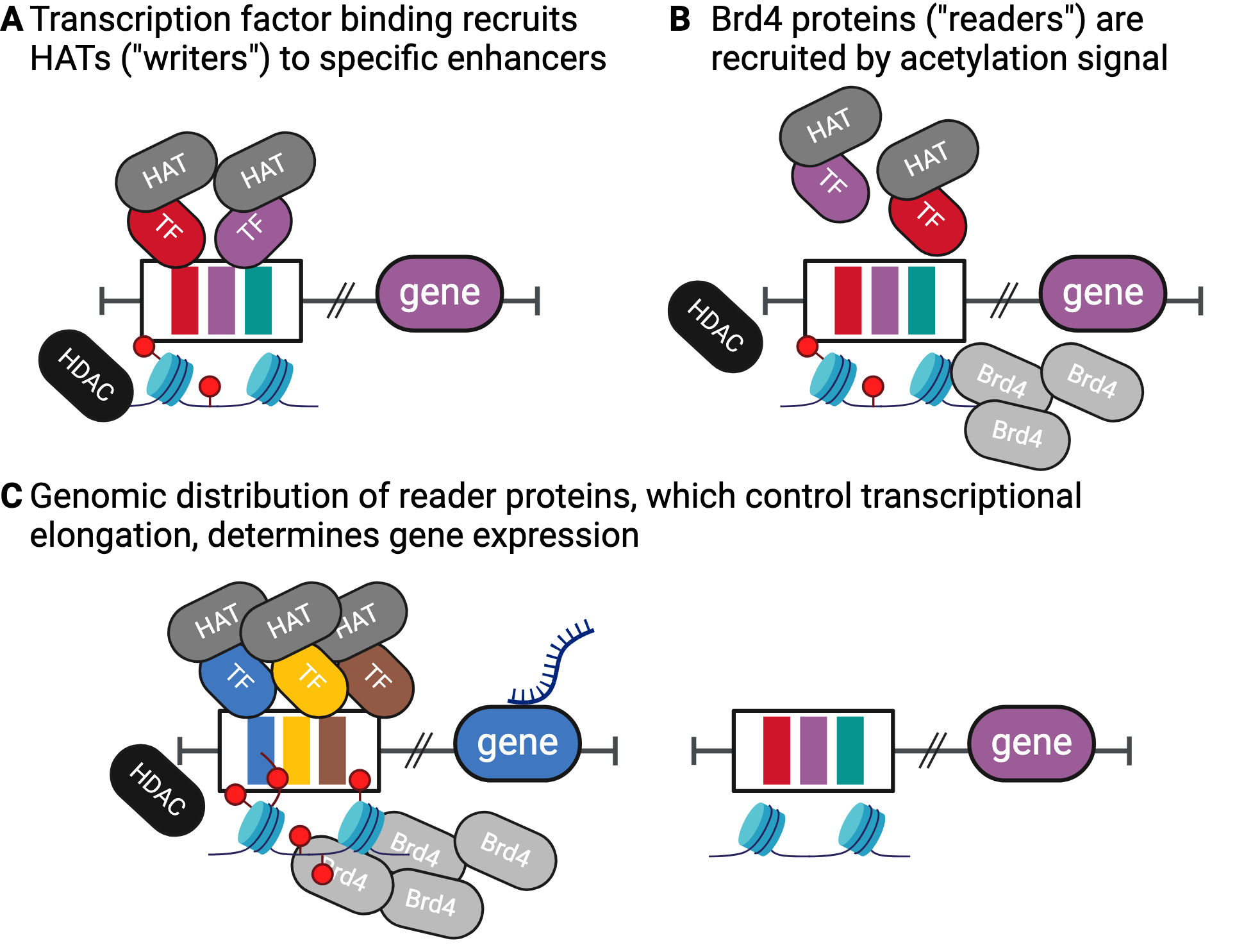}
\caption{\textbf{Mechanism of enhancer-based gene regulation.} (A) Transcription factor binding recruits HAT molecules ("epigenetic writers") to acetylate histone sites on enhancers, which are turned over by HDAC molecules ("epigenetic erasers"). (B) Acetylated histones recruit Brd4 molecules ("epigenetic readers"), (C) Acetylation facilitates competition between enhancers over epigenetic reader recruitment and drives gene expression changes.}
\label{fig:mechanism}
\end{figure}

We first introduce a mathematical framework designed to model gene expression in the context of enhancer coupling, based on the \emph{EnhancerNet} framework introduced in Karin \cite{karin2024enhancernet}. This framework focuses on the combined activities of transcription factors and histone acetylation-associated cofactors, although its principles are general and could be extended to other regulatory mechanisms. Additionally, while our study focuses on experimental results from the analysis of mammalian cells, the underlying mechanisms discussed appear to play a similar role across metazoa \cite{newman2020cell}.
Our analysis considers a set of $K$ enhancers ($i=1,\dots,K$) that regulate the expression of $N$ "core" transcription factors, denoted $x_1,\dots,x_N$, which, in turn, can bind any of these enhancers. These enhancers constitute a subset of all enhancers in the cell, as enhancers controlling genes that are not transcription factors represent outputs, and their exclusion does not affect model dynamics (Appendix). Additionally, we note that genes activated in a broadly cell-type-independent manner ("housekeeping genes") are expressed through mechanisms that are largely independent of enhancer acetylation and are less sensitive to inhibition of p300/CBP or Brd4 \cite{loven2013selective, di2014control,lipinski2020kat3,narita2021enhancers,neumayr2022differential}. We also consider $L$ transcription factors that act as downstream effectors of signalling pathways $y_1,\dots,y_L$. The model is high dimensional, with hundreds of relevant transcription factors, and thousands of relevant enhancer sites.

The mechanism of enhancer histone acetylation involves the recruitment of HATs to specific genomic sites by transcription factors. Transcription factors possess DNA-binding domains that recognize and bind to enhancer DNA sequences, as well as activation domains that interact with transcription factor-binding domains on HATs \cite{ferrie2024p300,inge2024rapid}. This interaction facilitates the recruitment of HATs to the enhancer region, where the chromatin-interacting domain of the HAT catalyzes histone acetylation \cite{narita2023acetylation,kikuchi2023epigenetic,ferrie2024p300}. Histones can be acetylated by HATs on many sites, and brodomain readers respond sensitively to histone hyperacetylation \cite{filippakopoulos2012histone,kikuchi2023epigenetic,narita2023acetylation,weinert2018time}. This is captured by our model (Appendix) whereby the overall histone acetylation on enhancer $i$, denoted, $m_i$, increases with HAT recruitment, which itself depends on binding of transcription factors to the enhancer, 
\begin{equation}
    \dot{m_i} = \kappa_1\left(\sum_{j=1}^{N}{\xi_{i,j}x_j} + \sum_{j=1}^{N}{\upsilon_{i,j}y_j} \right)-\kappa_{-1}m_i
\end{equation}

where $\kappa_1$ is proportional to active HAT molecules, $\kappa_{-1}$ is proportional to active HDAC molecules, $\xi_{i,j}$ denotes the affinity of transcription factor $x_j$ to the enhancer $i$, and $\upsilon_{i,j}$ denotes the affinity of signalling transcription factor $y_j$ to enhancer $i$. Finally, we note that the turnover associated with $m_i$ (determined by $\kappa_{-1}$) is rapid with a timescale of minutes \cite{narita2021enhancers}.

To facilitate transcription, acetylated histones recruit bromodomain proteins, including Brd4, which redistributes on the genome on the timescale of minutes \cite{dey2003double,brown2014nf,schmidt2015acute}. At the timescale of cell identity changes (hours-days), their distribution on relevant enhancer sites is proportional to
\begin{equation}
    p_i = \frac{e^{-f_i}}{\sum_k e^{-f_k}}
\end{equation}
where $f_k$ is the energy associated with Brd4 binding at site $k$. Since Brd4 binding depends sensitively on the degree of histone acetylation \cite{filippakopoulos2012histone,jung2014affinity,slaughter2021hdac}, we consider a first order approximation for the binding energy, $f_k=-\alpha m_k-b_k$, where $b_k$ corresponds to the baseline binding energy of a given site in the absence of acetylation, noting that $b_k$ may be affected by the presence of other chromatin or DNA modifications.  

Finally, the distribution of bromodomain proteins plays a critical role in regulating the expression of enhancer-associated cell identity genes through the mechanism of transcriptional pause-release and the recruitment of Mediator complex \cite{liu2013brd4,bhagwat2016bet,altendorfer2022brd4}. Both genetic deletion and pharmacological inhibition of these proteins lead to dramatic alterations in the expression of cell identity genes and block the expression of key cell identity transcription factors   \cite{di2014control,lee2017brd4,najafova2017brd4,arnold2021brd4,bressin2023high,narita2021enhancers}. As such, we consider the expression dynamics of the core transcription factors as driven by this distribution (Appendix),
\begin{equation}
    \dot{x_j}=\sum_{i=1}^{K}{q_{ij}p_i }-x_j = \sum_{i=1}^{K}{q_{ij}\frac{e^{\alpha m_i+b_i}}{\sum_k e^{\alpha m_k+b_k}} }-x_j 
\end{equation}
where $q_{i,j}$ corresponds to the contribution of enhancer $i$ to the expression of $x_j$, and time is rescaled to unity. Therefore, on the timescale of cell identity changes, gene expression changes of the regulatory network are given by
\begin{equation}
    \begin{split}
        \dot{\textbf{x}} &= \textbf{Q}^T \text{softmax} \left(\beta \Xi \textbf{x} + \beta \textbf{U}\textbf{y}+\textbf{b} \right) - \textbf{x} \\ &= \textbf{Q}^T \text{softmax} \left(\beta \Xi \textbf{x} + \textbf{w} \right) - \textbf{x}  \\&=g(\textbf{x},\beta,\textbf{Q},\Xi,\textbf{w})
    \end{split}
    \label{eq:everything}
\end{equation}
denoting $\beta=\frac{\kappa_1 \alpha}{\kappa_{-1}}$ as the effective "inverse-temperature-like" parameter, $\textbf{Q},\Xi,\textbf{U}$ are the corresponding rate matrices for $q_{i,j},\xi_{i,j},\upsilon_{i,j}$, the vectors $\textbf{x},\textbf{y}$ are the corresponding expression levels of the core and signalling transcription factors $x_j,y_j$, and $\textbf{b}$ is a vector whose entries correspond to $b_i$. We also set $\textbf{w}= \beta \textbf{U}\textbf{y}+\textbf{b}$ as the overall activation of enhancers that is not driven by the core transcription factors. We note that due to competition over HAT proteins we do not expect the dynamics to be sensitive to overall recruitment of bromodomain proteins to enhancers (Appendix). Importantly, while the derivation was based on p300/CBP mediated acetylation and Brd4 recruitment, the dynamics described by Eq.~\ref{eq:everything} are general and capture essential aspects of transcription factor-mediated epigenetic writer recruitment and epigenetic reader competition. 

The dynamics described in Eq.~\ref{eq:everything} involve the activation of enhancers by signalling and core transcription factors, followed by competition among enhancers for the recruitment of epigenetic reader molecules. When signalling pathways activate specific enhancers, they induce hyperacetylation, which drives changes in gene expression by redistributing epigenetic readers, as demonstrated experimentally \cite{brown2014nf,schmidt2015acute,marie2018hdac}. The competitive strength between enhancers can be modulated by altering the activity of epigenetic readers or erasers, which influence the parameter $\beta$. For instance, inhibiting HDAC activity (reducing $\kappa_{-1}$) leads to an overall increase in Brd4 binding to genomic sites \cite{marie2018hdac,slaughter2021hdac}, and, within our modelling framework, is predicted to elevate $\beta$. This enhances the relative competitiveness of different enhancers, potentially driving further changes in gene expression that could, in turn, modulate enhancer activity in a feedback loop.

Eq.~\ref{eq:everything} is general, and it is unclear how it results in cell identity behavior.  In \cite{karin2024enhancernet} we simplified it using a crucial experimental observation - that, across cell types, the transcription factors that control cell identity form dense autoregulatory networks where they co-bind adjacent enhancers \cite{hnisz2013super,saint2016models}, as illustrated in Figure~\ref{fig:crc_coarse}A. These enhancers, which have very similar binding profiles, will be co-activated by the core transcription factors in the same cellular context, and, therefore, we can coarse-grain them into a single "enhancer type". Since these enhancers activate the same transcription factors that co-bind them, after coarse-graining the rows of $\textbf{Q}$ and $\Xi$ will be correlated, and, therefore, Eq.~\ref{eq:everything} can be approximated by a symmetric form:
\begin{equation}
    \dot{\textbf{x}} = \Xi^T \text{softmax}\left(\beta \Xi \textbf{x} + \textbf{w} \right) - \textbf{x}
    \label{eq:enhancernetsym}
\end{equation}
The symmetric dynamics are associated with gradient descent along an energy (potential) function:
\begin{equation}
           \dot{\textbf{x}}=-\nabla V(\textbf{x},\beta, \Xi,\textbf{w})
           \label{eq:energy}
\end{equation}
with $V(\textbf{x},\beta,\Xi,\textbf{w})=-\beta^{-1}\log{Z } + \frac{1}{2} \textbf{x}^T\textbf{x}$, where $Z$ is the normalising constant of $\text{softmax}$ in Eq.~\ref{eq:enhancernetsym}. 

\section*{Results}
  
\subsection*{Emergence of hierarchical cell identity in transcription factor-enhancer interaction networks}

Eq.~\ref{eq:enhancernetsym} recapitulates a key aspect of cell identity specification, namely, at sufficiently large $\beta$, each of the observed cell identity profiles is indeed a stable configuration (attractor state) of the gene regulatory network. As such, the matrix $\Xi$ can be approximated from observed gene expression profiles. Using this initialisation, Eq.~~\ref{eq:enhancernetsym} makes a range of predictions on cell identity dynamics that do not require fitting unobserved parameters. One prediction regards the effect of transcription factor over-expression, where it correctly identifies the reprogramming effect of established recipes \cite{karin2024enhancernet}. The other prediction concerns hierarchical differentiation. A decrease in $\beta$, followed by a slow increase in $\beta$, transitions gene expression dynamics through a bifurcating landscape. When initialised with the target expression patterns of 11 major blood lineages, this process appears to recapitulate the blood progenitor landscape, including alternative differentiation trajectories. These observations suggest that changes in $\beta$ in Eq.~\ref{eq:enhancernetsym} may underlie progenitor states in hierarchical differentiation landscapes.

A regulatory mechanism supporting hierarchical cell identity must achieve two critical objectives: first, it must enable the stable maintenance of "coarse-grained" cellular identities (each linked to multiple "fine-grained" subtypes), and second, it must facilitate transitions between these states. This includes transitions toward "fine-grained" identities (as in differentiation) and reversions to "coarse-grained" identities, such as during shifts from stem cells to multipotent progenitor states or during de-differentiation in tissue repair. This dual capacity is not an inherent feature of dynamical systems, and it is particularly striking when considering the highly interconnected architecture of transcriptional networks, which generate combinatorial cell identities.

To understand how hierarchical regulation emerges from the dynamics of Eq.~\ref{eq:enhancernetsym}, we observe that as $\beta$ decreases, competition between enhancers relaxes, supporting the co-activation of enhancers with more distinct binding patterns that activate overlapping transcriptional programs. The intuition behind our analysis is illustrated in Figure~\ref{fig:crc_coarse}B. We first coarse-grain the enhancers in autoregulatory motifs into enhancer types, as they are always co-activated in the same cellular context. When $\beta$ is sufficiently large, competition between enhancer types is strong, and, therefore, enhancer types will uniquely determine cell identity. At lower levels of $\beta$, different enhancer types may be co-activated, and, by activating overlapping gene expression programs, they may support each other's activity.  Remarkably, we will show that in this case, different enhancer types coarse-grain and act as effective enhancer types that drive an equivalent coarse-grained expression program, resulting in hierarchical cell identity.

\begin{figure}%[tbhp]
\centering
\includegraphics[width=0.75\linewidth]{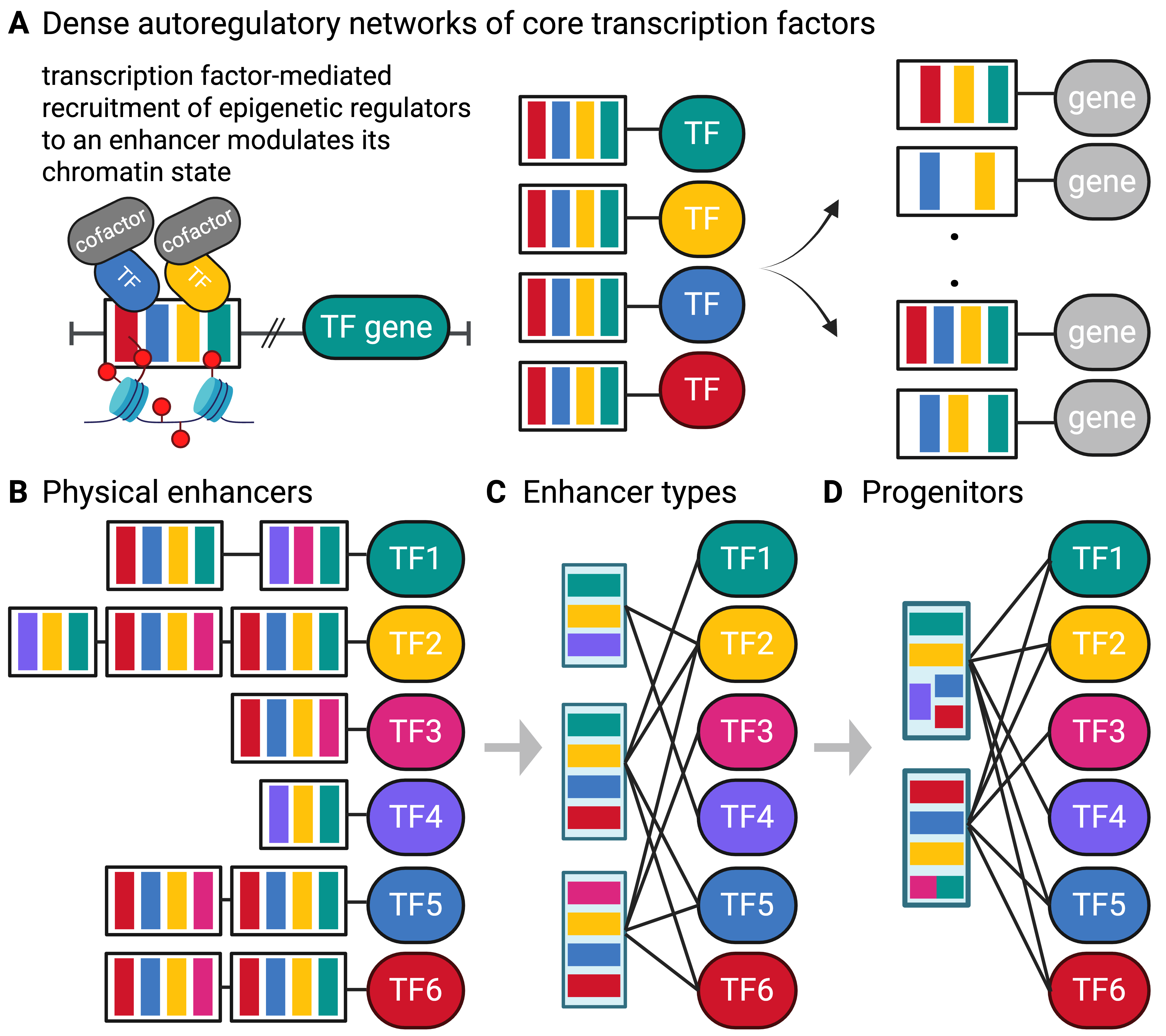}
\caption{\textbf{Scheme of process underlying emergence of hierarchical cell identity. }(A) Autoregulatory motifs in cell identity networks are made of enhancers with similar binding profiles that are co-bound by the genes they activate, (B-D) Coarsening of autoregulatory motifs (B) into enhancer types at high level of enhancer competition (C), and further coarsening into progenitors at lower level of enhancer competition (D). }
\label{fig:crc_coarse}
\end{figure}

To make this quantitatively precise, we show that specific weighted averages of terminal expression programs act as effective attractor programs at intermediate values of $\beta$ through self-similarity of the dynamics of Eq.~\ref{eq:everything}. The analysis hereafter considers patterns (rows of $\Xi$) of magnitude unity; a straightforward generalisation to patterns of unequal magnitude is provided in the Appendix. To see how progenitor patterns emerge from  Eq.~\ref{eq:everything}, consider a subset of indices $\mathcal{L}\subseteq\{1\dots K\}$ and denote the transformation $\textbf{Q}\xrightarrow{\mathcal{L}}\textbf{Q}',\Xi\xrightarrow{\mathcal{L}}\Xi',\textbf{w}\xrightarrow{\mathcal{L}} \textbf{w}'$, where $\textbf{Q}'=\textbf{Q}_{i\notin\mathcal{L}}\cup\{Q^*_{\mathcal{L}}\}, \;\Xi'=\Xi_{i\notin\mathcal{L}}\cup\{\Xi^*_{\mathcal{L}}\},\;\textbf{w}'=\textbf{w}_{i\notin\mathcal{L}}\cup\{w^*_{\mathcal{L}}\}$, with:
\begin{equation}
\label{eq:selfsimilarity}
\begin{split}
    Q^*_{\mathcal{L}}&=\frac{\sum_{i\in\mathcal{L}} Q_i e^{w_i }}{\sum_{i\in\mathcal{L}}{e^{w_i}}}  \\
    \Xi^*_{\mathcal{L}}&= \frac{\sum_{i\in\mathcal{L}} \Xi_i e^{w_i }}{\sum_{i\in\mathcal{L}}{e^{w_i}}} \\
    w^*_{\mathcal{L}}&=\log{\sum_{k\in\mathcal{L}}{e^{w_k}}}
\end{split}
\end{equation}
where $Q_i$ is the $i$-th row of $\textbf{Q}$, $\Xi_i$ is the $i$-th row of $\Xi$, and $w_i$ is the $i$-th element of $\textbf{w}$. The transformation coarse-grains the dynamics as the patterns and weights associated with the transformed dynamics are weighted averages of the original patterns. We refer to $ \Xi^*_{\mathcal{L}}$ as the \emph{progenitor pattern} associated with $\mathcal{L}$. We refer to the transformed dynamics as $g_{\mathcal{L}}(\textbf{x},\beta,\textbf{Q},\Xi,\textbf{w})$ and refer to the transformed potential (if it exists) as $V_{\mathcal{L}}(\textbf{x},\beta,\Xi,\textbf{w})$. Note that if the rows $\Xi_i, i\in\mathcal{L}$ are identical then the transformation leaves the dynamics unchanged; that is, when $\Xi_{i}=\Xi_{k}$ for $i,k\in \mathcal{L}$ then $g_{\mathcal{L}}(\textbf{x},\beta,\textbf{Q},\Xi,\textbf{w})=g(\textbf{x},\beta,\textbf{Q},\Xi,\textbf{w})$. This observation explains the emergence of symmetric dynamics due to the auto-regulatory motif described in the previous section (Figure~\ref{fig:crc_coarse}B,C), and we will therefore assume that the "baseline" dynamics are captured by the symmetric form Eq.~\ref{eq:enhancernetsym}, and note that  the transformation preserves symmetry, that is, if $\Xi_i=Q_i$ for all $i\in\mathcal{L}$ then $ Q^*_{\mathcal{L}}= \Xi^*_{\mathcal{L}}$.  

The coarse-graining  of a subset of expression patterns into a single progenitor pattern is associated with a $\beta$ threshold denoted $\beta_{\mathcal{L}}$. In the symmetric case, where the patterns are isolated and are of equal angle $\mu$ with each other, they destabilise the progenitor pattern $   \Xi^*_{\mathcal{L}}$ at a critical $\beta = 2 |\mathcal{L}| \mu^{-2}$.  More generally, for patterns of unequal angle, the coarse-grained dynamics are similar to the original dynamics up to a correction term, 
 \begin{equation}
     V_{\mathcal{L}}(\textbf{x},\beta,\Xi,\textbf{w}) \approx V(\textbf{x},\beta,\Xi',\textbf{w}') +  O(\beta \mu_{\mathcal{L}}^2)
 \end{equation}
where the correction term is at most $\frac{1}{4} \beta \mu_{\mathcal{L}} ^2$. The progenitor $  \Xi^*_{\mathcal{L}}$ thus has a \emph{persistence length} which depends on the statistical structure of the associated patterns $\Xi_{i\in\mathcal{L}}$ and is of order $ \beta_{\mathcal{L}}\approx \mu_{\mathcal{L}}^{-2}$. Thus, as $\beta$ decreases, subsets of enhancer types activating overlapping transcriptional programs may become effective enhancers of the dynamics (Figure~\ref{fig:crc_coarse}D, Figure~\ref{fig:coarsening}), with the transition occurring at larger $\beta$ for enhancer types with more overlapping target programs.

Importantly, the coarse-graining transformation can be applied to overlapping sets, whose persistence lengths can be distinct. For example, two overlapping pairs of patterns may be similar $\angle ({\Xi_i}, {\Xi_k})=\mu,\angle ({\Xi_i}, {\Xi_l})=\mu$,  while the third pair is less similar $\angle ({\Xi_l}, {\Xi_k})=2\mu$ , and the persistence length of the progenitor of ${\Xi_l,\Xi_k}$ is much shorter compared with the progenitors of $\Xi_i,\Xi_k$ and $\Xi_i,\Xi_l$. Finally, while the persistence length can capture precisely patterns of stabilisation and destabilisation of progenitors in strictly hierarchical sets of patterns with strong separation in  pattern similarities, in other cases more complex stability properties can emerge. However, by Eq.~\ref{eq:selfsimilarity}, a subset of patterns $\mathcal{L}$ become coupled through their progenitor $\Xi_{\mathcal{L}}^*$ for values of $\beta$ smaller than their persistence length.

\begin{figure}[H]%[tbhp]
\centering
\includegraphics[width=0.7 \linewidth]{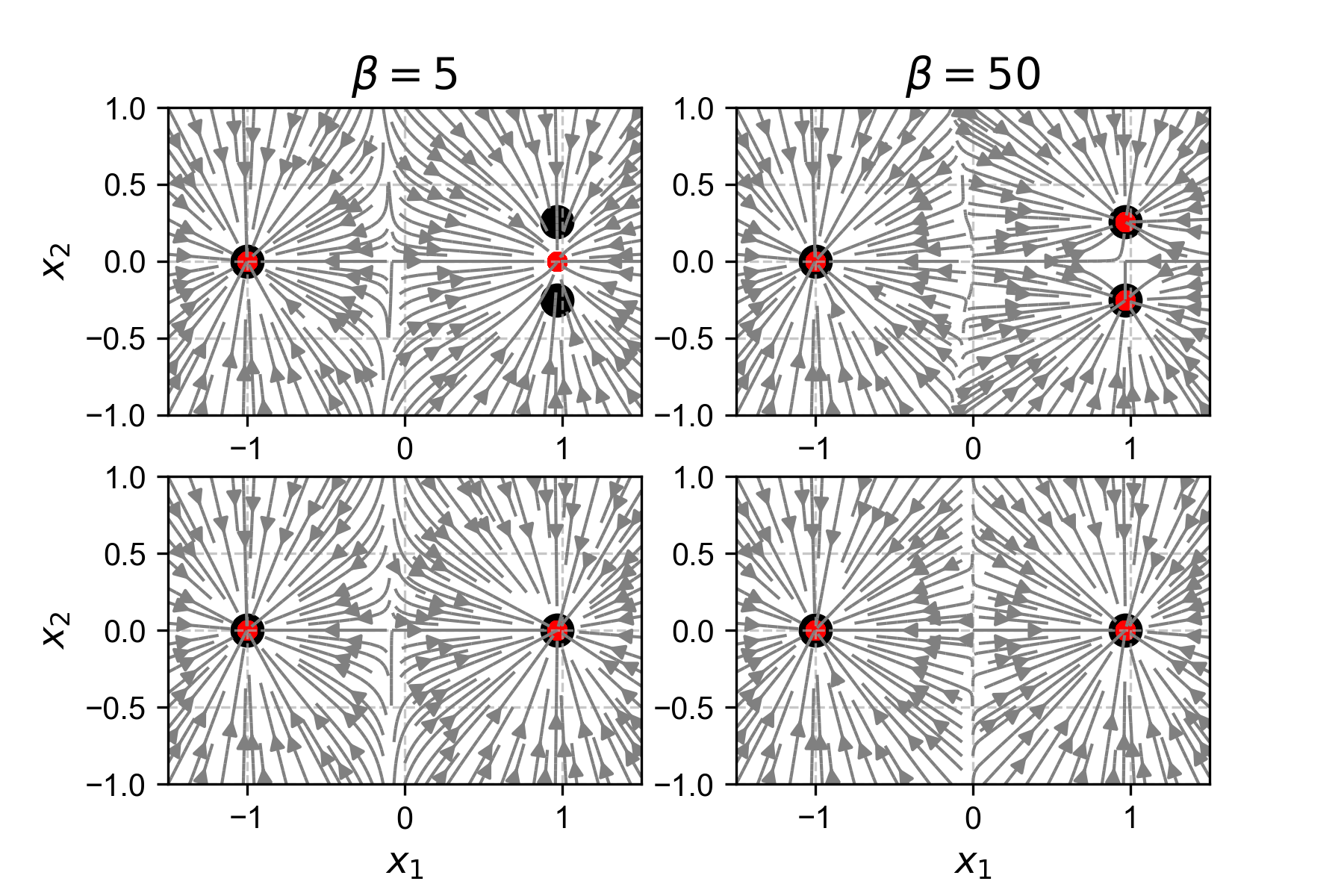}
\caption{ \textbf{Hierarchical coarsening of dynamics.} (A) Stream plots, patterns (black) and stable steady-states (red) for the dynamics given by Eq.~\ref{eq:enhancernetsym} and parameterized with $\textbf{w}=\textbf{0}$ (upper panels) or the dynamics where the two patterns associated with positive $x_1$ are coarse-grained by Eq.~\ref{eq:selfsimilarity} (lower panels). Coarse-graining is accurate for $\beta\ll\beta_{\text{crit}}$ (here $\beta_{\text{crit}} = 15$) but fails for $\beta>\beta_{\text{crit}}$. } 
\label{fig:coarsening}
\end{figure}

\subsection*{Hierarchical control of cell identity transitions}

The analysis so far indicates that different expression patterns, specified by autoregulatory circuits, as well as their weighted combinations, act as effective patterns of the dynamics for different values of $\beta$. The dependence on the maximal angle in the subset suggests that lower values of $\beta$ can support the stable expression of larger subsets of more distinct expression patterns. We now show that the weight vector $\textbf{w}$, which is modulated by signalling, plays a crucial role in determining which expression patterns achieve this stability.

In general, under the dynamics of Eq.~\ref{eq:enhancernetsym}, an increase of $w_i$ (the associated $\textbf{w}$ entry of pattern $i$) results in an increase in the basin of attraction of pattern $i$, at the expense of competing patterns. As an example, consider two stable competing patterns $\Xi_i, \Xi_k$ with an angle $\angle ({\Xi_i}, {\Xi_k})=\mu$. When $w_i$ increases to a sufficient extent, a bifurcation will occur where one of the competing patterns becomes destabilised, resulting in reprogramming. Thus, expression patterns associated with larger $\textbf{w}$ are more likely to be stable cell types.

While the weight vector entries of terminal and progenitor patterns appear mathematically indistinguishable, their biological interpretations are distinct. The weight entries of terminal patterns are directly interpretable within the physical framework of Eq.~\ref{eq:everything}, reflecting a baseline level associated with persistent chromatin and DNA modifications, as well as the binding of signaling transcription factors upstream of the cell identity network. In contrast, the weight entries of progenitor patterns are log-sum-exp transformations of the weights of their constitutive patterns. This distinction has two important implications for our analysis. The first is that progenitor patterns always have larger weights than their constitutive patterns, with progenitors of larger subsets having larger associated $\textbf{w}$ entries. The second is that increasing (or decreasing) the associated $w_i$ of a physical enhancer propagates upwards with both (a) increasing (resp. decreasing) the $\textbf{w}$ entry of its associated progenitor states, and (b) increasing (resp. decreasing) the contribution of its associated expression pattern  $\Xi_i$ to its progenitors. Crucially, this effect is inherently modular, with other progenitor states remaining unperturbed despite potentially sharing overlapping expression programs.

The modulation of $w_i$ thus controls both progenitor identity and the basin of attraction of $\Xi_i$. While the basin of attraction of $\Xi_i$ increases with $w_i$, that is, $\Xi_i$ is attractive from a wider range of states, the identities of progenitors associated with $\Xi_i$ become more similar to $\Xi_i$. Thus, modulation of $w_i$ changes the likelihood that a progenitor of $\Xi_i$ will be within its basin of attraction. Our analysis similarly extends to progenitor states, as an increase in $w_i$ leads to an increase in the basins of attractions of all progenitor patterns, due to an increase in their associated weights, and increases their similarity compared with competing patterns that do not incorporate $\Xi_i$. 

The observation on how $\textbf{w}$ controls progenitor identity and basins of attraction suggests that an annealing mechanism, whereby $\beta$ is decreased, and then slowly increased, can be used to transition between attractor states (Figure~\ref{fig:annealing}). Namely, modulation of $\textbf{w}$, such as by signalling, can control the end states of the annealing process. In the presence of noise or variability in parameters, an increase in $w_i$ can increase the \emph{rate} at which pattern $\Xi_i$ is formed. A feedback controller that adjusts $w_i$ according to the production rate of different patterns can therefore control the fractions of patterns produced.

\begin{figure}[H]
\centering
\includegraphics[width=1 \linewidth]{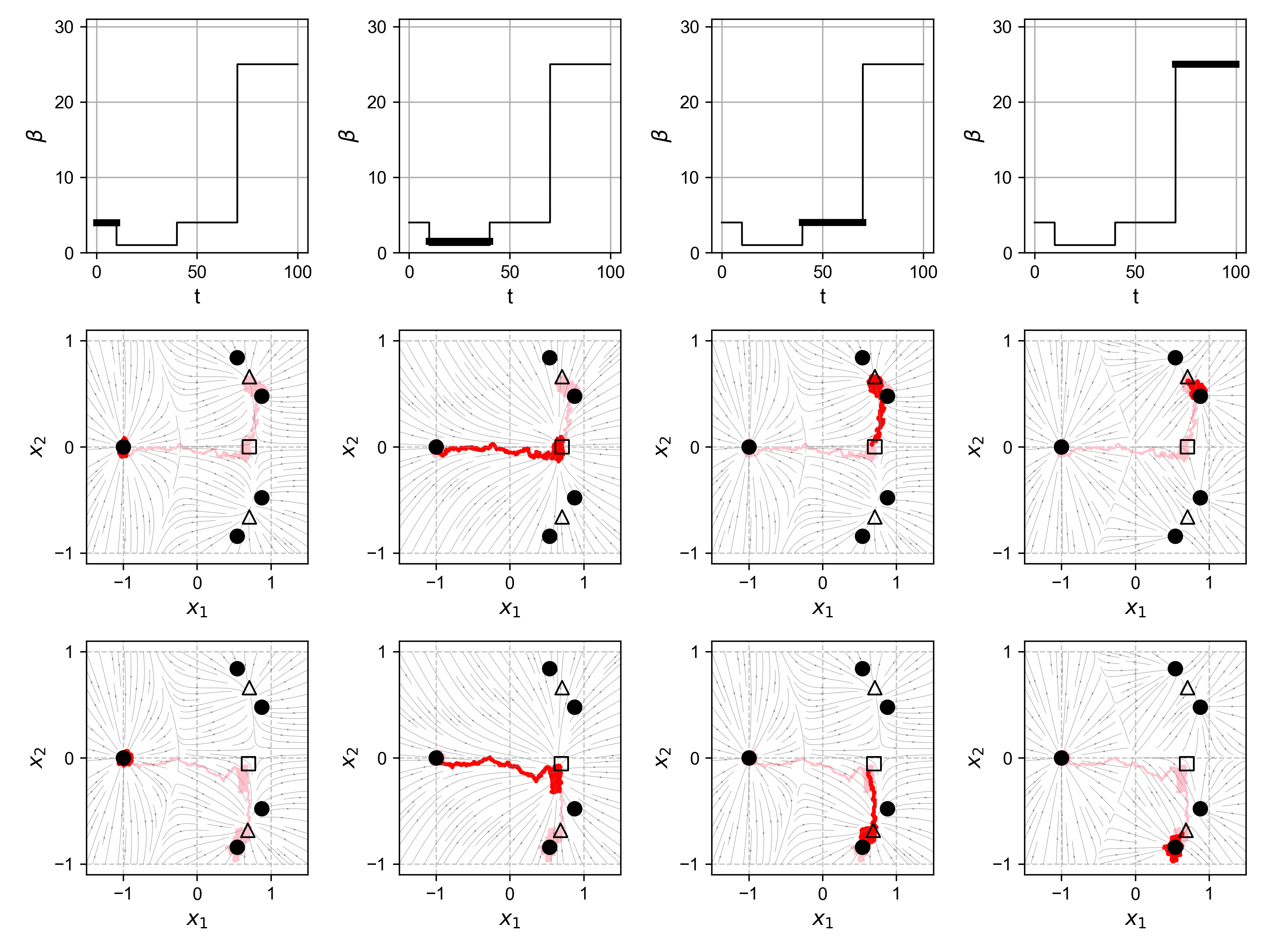}
\caption{ \textbf{Hierarchical coarsening of dynamics allows for control of transitions through progenitor states.} Enhancer transitions occur by a decrease followed by an increase in enhancer competition ($\beta$). Upper panel: $\beta$ input. Middle and lower panels: corresponding simulations of the model with the addition of noise to allow for variation in outcomes (Appendix). Bold red lines denote simulation outcome at the time points corresponding to bold input segments. Patterns (rows of $\Xi$) are marked as black circles, predicted bi-progenitor states are marked as empty triangles, and the predicted di-progenitor state is marked as an empty square. Stream plots of the correspondign $\beta$ value are plotted in the background of each panel. Middle panels: $\textbf{w}=0$, lower panels, $\textbf{w}=\{0,0,0,0.25\}$, with the last entry corresponding to the lowest pattern.} 
\label{fig:annealing}
\end{figure}

\subsection*{Emergence of hierarchical state transitions in hematopoiesis}

We now apply the framework to study the origin of progenitor states in blood formation.  Blood formation, or hematopoiesis, occurs through the differentiation of a stem cell population capable of generating all blood lineages. This process follows a well-defined hierarchical structure, starting with stem cells transitioning into multipotent progenitors, which then differentiate into more restricted cell types, ultimately giving rise to one of 11 major lineages. The potential of blood progenitors and their differentiation trajectories have been extensively characterized using various experimental approaches, including in vitro culture assays, transplantation studies, single-cell gene expression and chromatin analysis, and fate mapping \cite{ceredig2009models,laurenti2018haematopoietic}. Progenitor states can exhibit self-renewal and longevity, suggesting they represent stable attractors within the gene regulatory network. Furthermore, after acquiring their initial differentiated identity, many blood cell types undergo additional specialisation, such as during tissue migration or in response to immune challenges. 

A central question in hematopoiesis concerns the origin and nature of blood progenitor states. These states exhibit low-level expression of transcriptional programs associated with their descendant cell types \cite{hu1997multilineage, miyamoto2002myeloid, mercer2011multilineage, olsson2016single, zheng2018molecular, martin2021chromatin} and are responsive to differentiation signals. Importantly, different progenitors can possess partially overlapping fate potentials, leading to deviations from a strict hierarchical tree structure and enabling alternative differentiation routes to the same terminal fate \cite{ceredig2009models}. While, in theory, a vast number of progenitor states could exist (on the order of $2^{11}=2048$), far fewer have been observed. This discrepancy raises the question of why these specific sets of progenitors and differentiation trajectories are favored.

The model offers three key predictions to address this: (a) differentiation through multilineage priming reflects changes in $\beta$ that facilitate transitions between stem and differentiated states; (b) progenitor states with longer persistence lengths, compared to competing states, are more likely to be observed, as they act as attractor states across a broader range of $\beta$ values during the annealing process; and (c) among competing progenitor states with similar persistence lengths, those that encompass a wider range of terminal states are more likely to be observed due to their larger associated $\textbf{w}$ entries.

To test whether these predictions explain the observed progenitor states in blood formation, we analyzed gene expression data of transcription factors in blood lineages to compute the persistence lengths of all possible progenitor states (Figure~\ref{fig:blood}A). Observed progenitor states, as reviewed in \cite{ceredig2009models, brown2015lineage}, were represented as pie charts with black circumferences, reflecting their composition. Recently identified progenitors, such as the basophil/megakaryocyte/erythrocyte progenitors \cite{wanet2021cadherin, park2024ontogenesis} and the eosinophilic/neutrophilic progenitor \cite{metcalf2002stem, kim2017terminally, kwok2020combinatorial}, were highlighted with red circumferences, while unobserved states were marked as grey dots. Our analysis revealed that observed restricted-potential progenitors exhibit longer persistence lengths compared to their competitors, while at lower $\beta$ values, multipotent progenitor states dominate. This analysis correctly predicts the identity of the progenitor states.

We use blood as a model system to illustrate how signalling can regulate the relative production of terminal lineages by biasing progenitor states, thereby explaining observed effects such as erythropoietin's influence on blood differentiation. Our analysis also demonstrates how the transformation defined by Eq.\ref{eq:selfsimilarity} can predict progenitor expression patterns. Following the approach of Karin \cite{karin2024enhancernet}, we examine the dynamics of annealing for blood lineages in the presence of noise, which introduces variability in the annealing outcomes. Additionally, we define $\textbf{w}$ through proportional negative feedback, simulating the body's regulatory mechanisms that balance cellular production \cite{stanoev2022robust, simons2024tuning}. This results in a balanced differentiation landscape, where all fates are proportionally represented through a bifurcating process. Increasing $w_i$ enhances the production rate of lineage $i$ by amplifying the contribution of its expression pattern to progenitor states. 

To illustrate this, we increased the $\textbf{w}$ entry corresponding to the erythrocyte fate (Figure~\ref{fig:blood}C,D). This adjustment both increased the erythrocyte production rate and increased the relative contribution of the erythrocyte expression profile to its associated progenitor states, effectively recapitulating the known effects of erythropoietin on red blood cell production \cite{grover2014erythropoietin, eisele2022erythropoietin} (Figure~\ref{fig:blood}D). Such progenitor-biasing mechanisms may be a common feature in developing tissues \cite{soldatov2019spatiotemporal, moreau2021single, nusser2022developmental}. Finally, the identity of progenitor states and their persistence lengths are accurately captured by the coarse-graining transformation (Figure~\ref{fig:blood}E).

\begin{figure}[H]%[tbhp]
\centering
\includegraphics[width=\linewidth]{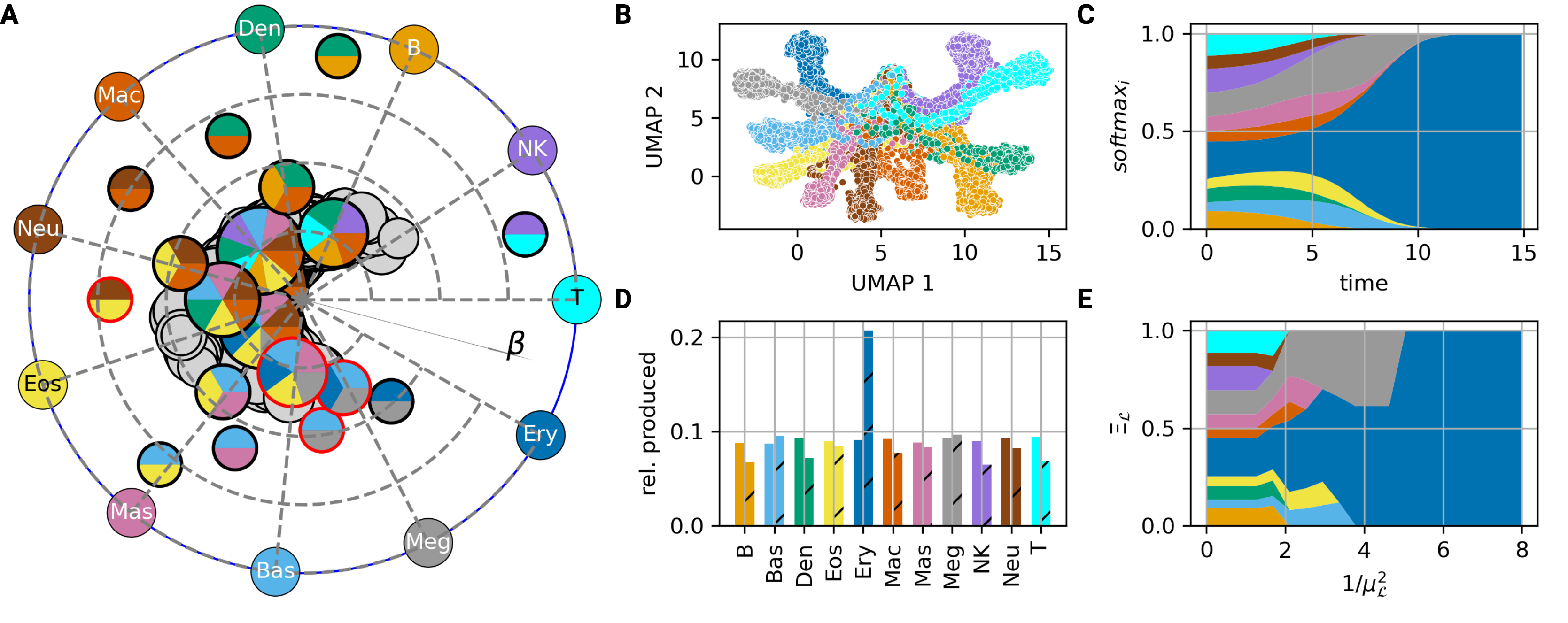}
\caption{\textbf{Hierarchical regulation of cell identity in blood formation.} (A) Progenitor states were plotted at a radius proportional to their corresponding persistence length, set by $1/\mu_{\mathcal{L}}^2$. (B-D) Simulations of annealing on hematopoietic lineages, with $\textbf{w}$ set by a proportional negative feedback controller with an error proportional to the imbalance of blood cell production. In all simulations the noisy version of Eq.~\ref{eq:enhancernetsym} was used (Appendix).  (B) Annealing trajectories were plotted using Uniform Manifold Approximation and Projection (UMAP) and colored by terminal cell type. (C) average contribution of different blood lineage expression programs to $\textbf{x}(t)$ during annealing that results in erythrocytes, given by $\text{softmax}\left(\beta \Xi \textbf{x} + \textbf{w} \right)$. (D) The relative production of different blood cell types, with the parameters inferred in (B-C) (blank) and with the $\textbf{w}$ entry of erythrocytes increased by 10\% (hatched), demonstrating how signalling can increase the relative production of different cell types in the annealing procedure.
(E) Heuristic prediction of progenitor states during an annealing process that results in erythrocyte differentiation. At each persistence length the largest progenitor state that is associated with an erythrocyte lineage is selected and plotted according to the relative contributions of each blood lineage to its expression profile, given by the transformation Eq.~\ref{eq:selfsimilarity}. The matrix $\Xi$ was initialised from the averaged mouse gene expression profiles of 11 major blood lineages from Haemopedia \cite{choi2019haemopedia}, initialised as in Karin \cite{karin2024enhancernet}. Simulation parameters are provided in the Appendix.}
\label{fig:blood}
\end{figure}

\subsection*{Pathways to dysregulation in cancer and implications for differentiation therapy}

The hierarchical cell identity model provides a framework for analyzing and predicting the dysregulation of cell identity in cancer. A hallmark of many cancers is the profound disruption of cellular identity, often driven by failures in differentiation processes \cite{hanahan2022hallmarks}. This phenomenon is particularly evident in the blood system, where numerous blood cancers arise from the inability of progenitor states to undergo proper maturation \cite{bispo2020epidemiology}.

Our analysis shows that, at a given level of $\beta$, progenitor states represent stable cell types with long-term persistence. Differentiation requires an increase in $\beta$ that results in appropriate cell fate bifurcations. Consequently, failures in differentiation can result from two distinct forms of dysregulation: either $\beta$ is maintained at an inappropriately low level, or the critical threshold of $\beta$ required for differentiation is elevated. Both mechanisms appear to play significant roles in cancer development.

The model predicts that inhibiting HDACs will increase $\beta$, while inhibiting HATs will decrease it. These predictions are consistent with experimental observations \cite{kang2004loss,iyer2004p300,tou2004regulation,pasqualucci2011inactivating,narita2021enhancers,likasitwatanakul2024chemical}. Broad inhibition of HDACs directly induces differentiation in cancers characterized by maturation failures, as well as in multipotent progenitors, resulting in cell cycle arrest and the upregulation of lineage-specific gene expression \cite{marks2000histone,gottlicher2001valproic,hsieh2004histone,san2019hdac,egolf2019lsd1,anastas2019re,xie2021targeting,likasitwatanakul2024chemical}. Within our modelling framework, HDAC inhibitors are predicted to decrease $\kappa_{-1}$ and therefore increase $\beta$, resulting in differentiation.

In addition to the dysregulation of epigenetic processes, many cancers—particularly blood cancers—are driven by mutations that disrupt the activity of transcription factors governing cell identity \cite{sive2014transcriptional}. These mutations often involve deletions or regulatory changes caused by chromosomal translocations. Consistent with the predictions of the model, transcription factor alterations linked to differentiation failures frequently occur in factors that are strongly and differentially expressed between cell fates. Notable examples include Pu.1 and Gata2, which are strongly and differentially expressed in myeloid lineages \cite{rosenbauer2004acute,hahn2011heritable,will2015minimal}, and Pax5, which is uniquely expressed in B cells \cite{mikkola2002reversion,cobaleda2007conversion}. Similarly, certain transcription factors are highly expressed in stem cells compared to differentiated lineages, and their dysregulation can lead to the aberrant stabilisation of a mixed stem-differentiated state. For instance, Lmo2 and Tal1, which are strongly expressed in hematopoietic stem cells but downregulated in lymphocytes, can drive lymphoblastic leukemias when mis-expressed \cite{mccormack2010lmo2,cleveland2013lmo2,sanda2017tal1}. Erg, another key transcription factor, is typically downregulated during stem cell differentiation, but its upregulation promotes a mixed stem/differentiated state in both myeloid \cite{carmichael2012hematopoietic,diffner2013activity} and lymphoblastic leukemias \cite{carmichael2012hematopoietic,thoms2011erg}.

Our analysis suggests that a critical mechanism underlying transformation following transcription factor dysregulation is a change in the persistence length of associated progenitor states, which follows from a larger overlap between their transcriptional profiles following mutation. This perspective provides a straightforward explanation for the effectiveness of HDAC inhibitors in differentiation therapy for blood cancers \cite{marks2000histone,gottlicher2001valproic,bots2014differentiation,san2019hdac,wang2020role}, as they may help restore normal differentiation by increasing $\beta$ further to facilitate differentiation.

\section*{Discussion}

In this study we developed a predictive mathematical framework to explain the mechanistic basis of progenitor states and hierarchical differentiation landscapes in animal gene regulatory networks. We demonstrate that hierarchical cell identity is an essential outcome of the epigenetic write/erase/read mechanisms that control animal gene expression. Our model provides specific predictions at multiple levels, including on the perturbation of epigenetic regulators, transcription factors, and enhancers. We identify cell types as emergent properties reflecting the balance of enhancer cooperation and competition and demonstrate how progenitor states can facilitate sensitive state transitions through progenitor biasing.

From a mathematical perspective, our analysis relates animal gene regulation through enhancer activation to models of associative memories in neural networks. Cell identity behavior in animals exhibits many qualitative similarities to associative memory networks \cite{lang2014epigenetic,boukacem2024waddington,lang2014epigenetic,pusuluri2017cellular,karin2024enhancernet}. Different cell types in an animal have unique gene expression profiles that remain stable throughout its life. These profiles correspond to the 'memory patterns' of the gene regulatory network. Dynamic retrieval of these expression patterns is crucial during differentiation processes, such as during development or in the self-renewal of adult tissues from stem populations. Our analysis identifies the gene regulatory networks that control cell identity with models known as Dense Associative Memories, which allow for the storage and retrieval of many memory patterns through higher order interactions \cite{krotov2016dense,demircigil2017model,krotov2020large,ramsauer2020hopfield,lucibello2024exponential}, providing essential functional benefits for cell identity regulation \cite{karin2024enhancernet}. Through this mathematical analogy, Eq.~\ref{eq:enhancernetsym} is relevant to a wide range of applications in neuroscience and machine learning \cite{ramsauer2020hopfield,ambrogioni2023search,spens2024generative}. Our discovery that the attractor landscape at intermediate $\beta$ is dominated by self-similar dynamics may thus have broad implications, including for understanding phenomena such as learning generalisation in neural networks \cite{pham2024memorization}.

Unlike classical gene regulatory models that presuppose hierarchical or modular network architectures \cite{bolouri2002modeling,betancur2010assembling,peter2011evolution}, our framework reveals hierarchy as an emergent property arising from fundamental enhancer dynamics. By modeling the interplay of enhancer cooperation and competition, we unify diverse phenomena—from combinatorial cell fate specification and multilineage priming to progenitor state plasticity and perturbation responses—without imposing any a priori hierarchical structure on the regulatory network. Our analysis also demonstrates that enhancer competition enables a highly modular control of differentiation, even in the absence of modularity at the level of the transcriptional network. Signalling pathways that activate a specific transcriptional program through its enhancers will selectively perturb the associated progenitor states without impacting other states. This finding is unexpected, given that transcription factors are widely shared across different cell types, resulting in densely interconnected regulatory networks at the transcription factor level. Modularity emerges only when regulation is examined through the lens of enhancer-transcription factor interactions. Our findings highlight that incorporating enhancer structure—whether implicitly through symmetry considerations, as in this study, or explicitly—is crucial for accurately modeling the regulation of cell identity.

Within our modelling framework, competing fates are resolved by global modulation of enhancer competition. From a functional point of view, this modulation is essential to allow for sensitive transitions between cellular states, such as stem and differentiated states, in a signalling-sensitive manner. While, in principle, many biological mechanisms may underlie enhancer competition, a key candidate mechanisms is differential recruitment of bromodomain proteins, including Brd4, by differential enhancer acetylation. Bromodomain proteins are specifically associated with the regulation of cell identity genes, and, as such, their modulation allows for perturbation of the cell identity network without dysregulating more ubiquitously expressed housekeeping genes. Lower $\beta$ states are thus predicted to be associated with broader enhancer-associated transcription due to broader sharing of transcriptional regulators between enhancers. 

Our model primarily examines the role of HDACs and HATs in the global modulation of enhancer competition and correctly predicts their effects on differentiation. However, additional pathways for this response are also possible. For instance, Myc proteins are important regulators of differentiation across various cell types, transitioning from high to low expression levels across many differentiation hierarchies. They exhibit anti-differentiation effects that may be independent of their pro-proliferation functions \cite{ryan1997cell,leon2009inhibition,bahr2018myc}. How do Myc proteins mediate this effect? Myc proteins has been demonstrated to bind promiscuously across regulatory sites, amplifying transcription \cite{sabo2014genome} and recruiting and sequestering coactivators such as p300/CBP \cite{faiola2005dual}. Within our modelling framework, by recruiting coactivators at competing regulatory sites, Myc can drive a reduction in $\beta$, thereby inhibiting differentiation.

The dynamics of differentiation in our model exhibit both similarities and differences compared to Waddington's classic epigenetic landscape picture of cellular differentiation \cite{waddington1957strategy,nimmo2015primed,boukacem2024waddington}. As in the Waddingtonian picture, differentiation in our model proceeds in a hierarchical manner through a series of bifurcations, from multipotent to ever more restricted progenitors. The modelling framework presented here makes the concept of potency quantitatively precise, in that we show that a differentiation signal for fate $i$, which modulates $w_i$, will perturb all progenitors associated with this fate towards it, while other progenitors will not be affected. Our analysis thus unifies the experimental observations on hierarchical differentiation, multilineage priming, and differentiation potency, into a single mathematical framework. In contrast to the classical Waddingtonian picture, however, our hierarchical picture does not entail tree-like differentiation landscape. Rather, analysis suggests that differentiation landscapes may be much more flexible, with multiple possible trajectories to the same target fate. Moreover, the trajectories taken may themselves change according to signalling which can lead to changes in the basins of attraction of the progenitor states. These flexible and non-tree like dynamics are in line with experimental knowledge on complex differentiation dynamics in systems such as blood formation and the neural crest \cite{ceredig2009models,erickson2023transcriptional}. While our analysis proposes flexible differential landscapes, these landscapes are nevertheless constrained by the statistical structure of the target states and can be quantitatively predicted through the persistence length measure.

%\bibliography{bib}% common bib file
%% if required, the content of .bbl file can be included here once bbl is generated
\providecommand{\noopsort}[1]{}\providecommand{\singleletter}[1]{#1}%
%% BioMed_Central_Bib_Style_v1.01

\end{document}

\maketitle
\appendix
\tableofcontents
\newpage

\section{Derivation of model for cell identity regulation}

\subsection{Enhancer acetylation}

The first step in the model involves the acetylation of histone proteins at a specific enhancer site $j$. Each regulatory transcription factor (TF) is associated with an expression level $z_j$ ($j\in 1\dots N$), corresponding to its protein concentration in the cell. Enhancer activation begins with the recruitment of histone acetyltransferases (HATs), such as p300 and CBP, to specific genomic sites, followed by the acetylation of nearby histones by these HATs \cite{ferrie2024p300}. We denote the concentration of HATs ($H$) as $h$ and model their recruitment to a genomic site $i$ as a two-step process involving the formation of a complex.

\begin{equation*}
    Z_j + H \underset{k_{-1}}{\stackrel{k_1}{\rightleftharpoons}} C_j
\end{equation*}
For $j\in 1\dots N$. This corresponds to an ODE system, 
\begin{equation*}
    \begin{split}
        \dot{c_j} &= k_1 z_j h - k_{-1}c_j, \;\; j\in 1\dots N \\
        \dot{h} &= k_{-1}\sum_j {c_j } - k_1 h \sum_j {z_j }
    \end{split}
\end{equation*}

these equations imply the conservation law, $h=h_0-\sum_jc_j$. Thus, at steady-state, we have that, for all $j$,
\begin{equation*}
    c_j = \frac{k_1}{k_{-1}} z_j h = \frac{k_1}{k_{-1}} z_j \left(h_0-\sum_jc_j \right)
\end{equation*}
summing over all $j$, and denoting $c_{\text{tot}}=\sum_jc_j, z_{\text{tot}}= \sum_j z_j$, we have,
\begin{equation*}
    \frac{c_{\text{tot}}}{h_0-c_{\text{tot}}} =  \frac{k_1}{k_{-1}} z_{\text{tot}}
\end{equation*}
and therefore,
\begin{equation*}
    c_{\text{tot}}= h_0 \frac{k_1}{k_{-1}} \frac{ z_{\text{tot}}}{\frac{k_1}{k_{-1}} z_{\text{tot}}+1}
\end{equation*}
and
\begin{equation*}
    h=h_0-c_{\text{tot}}= h_0 \left(1-\frac{k_1}{k_{-1}} \frac{ z_{\text{tot}}}{\frac{k_1}{k_{-1}} z_{\text{tot}}+1}\right) = \frac{h_0}{\frac{k_1}{k_{-1}} z_{\text{tot}}+1}
\end{equation*}
and
\begin{equation*}
    c_j= \frac{k_1}{k_{-1}} z_j \left(h_0-h_0 \frac{k_1}{k_{-1}} \frac{ z_{\text{tot}}}{\frac{k_1}{k_{-1}} z_{\text{tot}}+1}\right) = \frac{ \frac{k_1}{k_{-1}} h_0}{\frac{k_1}{k_{-1}} z_{\text{tot}}+1} z_j \approx \frac{h_0  }{z_{\text{tot}}} z_j
\end{equation*}
where in the last approximation we used the observation that the levels of HATs are considered limiting relative to transcription factor abundance \cite{giordano1999p300,gillespie2020absolute}, and, as such, we expect to observe a relatively small amount of free $h$ (that is, we are at the regime where $\frac{k_1}{k_{-1}} z_{\text{tot}}\gg1$).
The acetylation rate of site $i$  will then depend on its transcription factor binding profile. Denoting $M_i$ as level of acetylation of site $i$, and denoting $\zeta_{i,j}$  as the relative binding affinity of TF $j$ for site $i$, we have that:
\begin{equation*}
    M_i\xrightarrow{\sum_j{\zeta_{i,j}c_j}} M_i+1
\end{equation*}
Histone deacetylation, on the other hand, is driven by histone deacetylases (HDACs) which, unlike HATs, are highly abundant and have broad genomic localization \cite{bornelov2018nucleosome,gillespie2020absolute}. Denoting by $\kappa_{-1}$ the activity of HDACs (which is proportional to their abundance),
\begin{equation*}
       M_i\xrightarrow{\kappa_{-1}} \varnothing
\end{equation*}
resulting in the dynamics
\begin{equation*}
    \begin{split}
        \dot{m_i}&= \sum_j{\zeta_{i,j}c_j} - \kappa_{-1} m_i =  \sum_j{\zeta_{i,j} \frac{h_0  }{z_{\text{tot}}} z_j} - \kappa_{-1} m_i \\
        &= h_0 \sum_j{\zeta_{i,j} \frac{z_j}{z_{\text{tot}}}} - \kappa_{-1} m_i
    \end{split}
\end{equation*}

\subsection{Recruitment of epigenetic readers}

Enhancer acetylation plays a crucial role in recruiting bromodomain-containing proteins, such as Brd4, which are important regulators of transcription. Brd4 recognizes acetylated histones through its bromodomains and facilitates transcription by promoting RNA polymerase II (RNAPII) pause-release and interacting with the Mediator complex. Brd4 recruits positive transcription elongation factor b (P-TEFb) to phosphorylate RNAPII, enabling productive elongation. These proteins possess acetylation-binding domains that confer sensitivity to hyperacetylated states, enabling them to respond dynamically to changes in chromatin acetylation levels \cite{filippakopoulos2012histone,slaughter2021hdac}. Let $r_i$ denote the number of bromodmain proteins (e.g. Brd4) recruited to site $i$. This is given by
\begin{equation*}
    r_i = r_0 \frac{e^{-f_i}}{K_r+ \sum_{k}e^{-f_k}}
\end{equation*}
where $r_0$ is the total number of bromodomain proteins, $f_i$ is the energy associated with the binding of a bromodomain protein to site $i$, and $K_r$ represents the fraction of unbound bromodomain proteins (which may be otherwise bound to other sites in genome, or unbound to the genome).  Taking a first order approximation for the energy function, we have $f_k=-\alpha m_k-b_k$, where $b_k$ corresponds to the baseline binding energy of a given site in the absence of acetylation, noting that $b_k$ may be affected by the presence of other chromatin or DNA modifications.  

We can also consider the relative bromodomain protein distribution among the relevant enhancer sites, given by
\begin{equation*}
    p_i = \frac{r_i}{\sum_k r_k} = \frac{e^{-f_i}}{\sum_{k}e^{-f_k}}
\end{equation*}
Note that the number of bound bromodomain proteins is given by  $r_{\text{bound}}=\sum_k r_k = r_0 \frac{\sum_{k}e^{-f_k}}{K_r+ \sum_{k}e^{-f_k}}$. 

\subsection{Enhancer-mediated transcription}

We now estimate the rate of transcription of an (arbitrary) gene indexed $j$ mediated by an enhancer site $i$. This is proportional to  $q_{i,j} r_i$ where $q_{i,j} $ is the coupling of the site to the gene, which is related to the contact frequency between the sites. The overall transcriptional dynamics of a gene $z_j$ are given by
\begin{equation*}
   \tau^{-1} \dot{z_j} = \sum_{i}{q_{i,j} r_i} -\gamma z_j  = \sum_{i}{r_{\text{bound}} q_{i,j} p_i } - z_j
\end{equation*}
where $\tau^{-1}$ is the timescale of dynamics of transcription. As overall transcription of enhancer-regulated genes increases in proportion to $r_{\text{bound}}$, we can denote $\hat{z}_j=z_j r_{\text{bound}}$, and mark $\hat{z}_{\text{tot}}= \sum_k \hat{z}_k$. Then
\begin{equation*}
    \tau^{-1} \dot{\hat{z_j}} =  \sum_{i}{q_{i,j} p_i } - \hat{z}_j
\end{equation*}
and
\begin{equation*}
     \dot{m_i} = \frac{h_0}{\hat{z}_{\text{tot}}} \sum_j{\zeta_{i,j} \hat{z}_j} - \kappa_{-1} m_i = \kappa_{1} \sum_j{\zeta_{i,j} \hat{z}_j} - \kappa_{-1} m_i
\end{equation*}
where we denote $\kappa_{1}=\frac{h_0}{\hat{z}_{\text{tot}}}$. We will set $\beta=\alpha\frac{\kappa_1}{\kappa_{-1}}$. It is important to note that $\kappa_{1}$ is a dynamical quantity that changes with the overall scaled transcription factor levels $\hat{z}_{\text{tot}}$. However, within the context of this manuscript, it will be more useful to consider it to be constant, as its main effects will be to rescale the effective $\beta$ at the vicinity of various attractor patterns.

\subsection{Dynamics of transcriptional network}

We will partition the transcription factors $\hat{z}_1,\hat{z}_2,\dots$ into two sets. The first are transcription factors with constitutive binding. These are transcription factors whose activity patterns are dominated by whether they are transcribed. We denote these as $x_j$ ($j\in1,\dots,N$) and their binding affinities as $\xi_{i,j}$. The second are signaling dependent transcription factors, whose activity patterns are dominated by the signalling environment. As such we will assume that the latter are constitutively expressed within a given context (ignoring their enhancer-dependent production) and denote them as $y_j$ ($j\in1,\dots,L$) and their binding affinities as $\upsilon_{i,j}$.
We therefore have the dynamics:
\begin{equation}
    \begin{split}
     \tau^{-1} \dot{x_j} &=  \sum_{i}{q_{i,j} p_i } - x_j \\
     \dot{m_i} &=   \kappa_{1} \sum_{j=1}^N{\xi_{i,j} x_j}+ \kappa_{1} \sum_{j=1}^L{\upsilon_{i,j} y_j} - \kappa_{-1} m_i
\end{split}
\end{equation}

Since the turnover of histone acetylation is rapid, with a timescale of minutes, we can take the quasi-steady-state of $m_i$ as 
\begin{equation*}
    m_i = \frac{\kappa_1}{\kappa_{-1}} \left(\sum_{j=1}^N{\xi_{i,j} x_j}+ \kappa_{1} \sum_{j=1}^L{\upsilon_{i,j} y_j} \right)
\end{equation*}
set
\begin{equation*}
    \begin{split}
        p_i &= \frac{e^{-f_i}}{\sum_{k}e^{-f_k}} = \frac{e^{\alpha m_i+b_i}}{\sum_{k}e^{\alpha m_k+b_k}} \\ &= \frac{e^{\beta  \left(\sum_{j=1}^N{\xi_{i,j} x_j}+  \sum_{j=1}^L{\upsilon_{i,j} y_j} \right)+b_i}}{\sum_{k}e^{\beta  \left(\sum_{j=1}^N{\xi_{k,j} x_j}+ \sum_{j=1}^L{\upsilon_{k,j} y_j} \right)+b_k}}        
    \end{split}
\end{equation*}
which, when substituted into the dynamical equation for $x_j$, and rescaling time $\tau=1$, yields the dynamical system
\begin{equation}
        \dot{\textbf{x}} = \textbf{Q}^T \text{softmax} \left(\beta \Xi \textbf{x} + \beta \textbf{U}\textbf{y}+\textbf{b} \right) - \textbf{x}
\end{equation}
where the matrices $\textbf{Q},\Xi,\textbf{U}$ are the corresponding rate matrices for $q_{i,j},\xi_{i,j},\upsilon_{i,j}$, the vectors $\textbf{x},\textbf{y}$ are the corresponding expression levels of the core and signalling transcription factors $x_j,y_j$, and $\textbf{b}$ is a vector whose entries correspond to $b_i$. This can further simplified by setting $\textbf{w}= \beta \textbf{U}\textbf{y}+\textbf{b}$,
\begin{equation}
 \dot{\textbf{x}} = g(\textbf{x},\beta,\textbf{Q},\Xi,\textbf{w}) = \textbf{Q}^T \text{softmax}\left(\beta \Xi \textbf{x} + \textbf{w} \right) - \textbf{x}
 \label{eq:enhancernet}
\end{equation}

\section{Coarse-graining}
\subsection{Coarse-graining transformation for patterns of identical magnitude}
\newtheorem{proposition}{Proposition}[section]
\newtheorem{lemma}[proposition]{Lemma}
\par In this section, we present the coarse graining transformation for patterns of identical magnitude and inspect its mathematical properties. Suppose we have $N$ nodes (TFs in the gene regulatory network) and $K$ memory patterns.  We consider a subset of indices $\mathcal{L}\subseteq\{1\dots K\}$ and define our coarse-graining transformation for patterns of identical magnitude by $\textbf{Q}\xrightarrow{\mathcal{L}}\textbf{Q}',\Xi\xrightarrow{\mathcal{L}}\Xi',\textbf{w}\xrightarrow{\mathcal{L}} \textbf{w}'$, where:
\begin{equation}
 \textbf{Q}'=\textbf{Q}_{i\notin\mathcal{L}}\cup\{Q^*_{\mathcal{L}}\}, \;\Xi'=\Xi_{i\notin\mathcal{L}}\cup\{\Xi^*_{\mathcal{L}}\},\;\textbf{w}'=\textbf{w}_{i\notin\mathcal{L}}\cup\{w^*_{\mathcal{L}}\}
\end{equation}
with:
\begin{equation} \label{eqn:coarse_graining_transformations}
\begin{split}
 Q^*_{\mathcal{L}}&=\frac{\sum_{i\in\mathcal{L}} Q_i e^{w_i }}{\sum_{i\in\mathcal{L}}{e^{w_i}}} = \text{softmax}\left(\textbf{w}_{i\in\mathcal{L}}\right) Q_{i\in\mathcal{L}}\\
 \Xi^*_{\mathcal{L}}&= \frac{\sum_{i\in\mathcal{L}} \Xi_i e^{w_i }}{\sum_{i\in\mathcal{L}}{e^{w_i}}} = \text{softmax}\left(\textbf{w}_{i\in\mathcal{L}}\right) \Xi_{i\in\mathcal{L}}\\\
 w^*_{\mathcal{L}}&=\log{\sum_{k\in\mathcal{L}}{e^{w_k}}}
\end{split}
\end{equation}
This transforms the dimension of the matrices from $N \times K$ to $N \times K'$ with $K'=K-|\mathcal{L}|+1$. We denote the transformed dynamics as $g_{\mathcal{L}}(\textbf{x},\beta,\textbf{Q},\Xi,\textbf{w})$.

We will begin by showing that the order in which we apply the coarse-graining transformation in Eq. \ref{eqn:coarse_graining_transformations} does not matter, i.e. our transformation is in essence commutative. More precisely, consider two subsets of indices $\mathcal{L}_{1},\mathcal{L}_{2} \subseteq \{1,...,K\}$, such that $\mathcal{L}_{1} \cap \mathcal{L}_{2}= \varnothing $ and let $\mathcal{L}= \mathcal{L}_{1} \cup \mathcal{L}_{2}$.
We let $L_{1,Q}$ be the transformation corresponding to the transformation in Eq.~\ref{eqn:coarse_graining_transformations} above with indices given by $\mathcal{L}_{1}$ acting on the $\textbf{Q}$ matrices, and similarly for $\Xi$ and $\textbf{w}$.
We define $L_{1,Q} \odot L_{2,Q}$ as the transformation that reduces all rows given by indices in $\mathcal{L}_{2}$ to one row (as above) and then reducing all the indices in $\mathcal{L}_{1}$ (with the labelling corresponding to the original labelling) to one row. 
\begin{proposition}
When $\mathcal{L}_{1} \cap \mathcal{L}_{2}= \varnothing $, we have $L_{1,Q} \odot L_{2,Q}=L_{2,Q} \odot L_{1,Q}$ and similarly for transformations acting on $\Xi$ and $\textbf{w}$.
\end{proposition}
\begin{proof}
Note that as $\mathcal{L}_{1} \cap \mathcal{L}_{2} = \varnothing$,
$L_{1,Q} \odot L_{2,Q}(\textbf{Q})=L_{1,Q}(\textbf{Q}_{i\notin\mathcal{L}_{2}}\cup\{Q^*_{\mathcal{L}_{2}}\})=\textbf{Q}_{i\notin\mathcal{L}_{1}} \cup \textbf{Q}_{i\notin\mathcal{L}_{2}}\cup\{Q^*_{\mathcal{L}_{1}}\} \cup \{Q^*_{\mathcal{L}_{2}}\}$ and so
$$L_{1,Q} \odot L_{2,Q}=L_{2,Q} \odot L_{1,Q}.$$
Exactly the same arguments work for $L_{1,\textbf{w}}$ and $L_{1,\Xi}.$

\end{proof}

\par Now suppose as above that $\mathcal{L}_{1},\mathcal{L}_{2} \subseteq \{1,...,K\}$, such that $\mathcal{L}_{1} \cap \mathcal{L}_{2}= \varnothing $ and let $\mathcal{L}= \mathcal{L}_{1} \cup \mathcal{L}_{2}$. And let $L_{1},L_{2},L$ be the corresponding maps (acting on the same variables, i.e. either $\textbf{Q},\textbf{w}$ or $\Xi$), further define the map $L^{*}$ as the map associated with the transformation given by the indices $\mathcal{L}^{*}=\{1^{*},2^{*}\}$, where $1^{*},2^{*}$ are the indices to which $\mathcal{L}_{1}$ and $\mathcal{L}_{2}$ are mapped respectively. Then we have the following:
\begin{proposition}
 Suppose $\mathcal{L}_{1},\mathcal{L}_{2} \subseteq \{1,...,K\}$ such that $\mathcal{L}_{1} \cap \mathcal{L}_{2}= \varnothing $ and let $\mathcal{L}= \mathcal{L}_{1} \cup \mathcal{L}_{2}$. Then $L=L^{*}\odot L_{1} \odot L_{2}$ for $\textbf{Q},\Xi,\textbf{w}$.
 
\end{proposition}
\begin{proof}
 Consider $\textbf{Q}$, then from the previous proposition we know that $L_{1,Q} \odot L_{2,Q}(\textbf{Q})=\textbf{Q}_{i\notin\mathcal{L}_{1}} \cup \textbf{Q}_{i\notin\mathcal{L}_{2}}\cup\{Q^*_{\mathcal{L}_{1}}\} \cup \{Q^*_{\mathcal{L}_{2}}\}$. Therefore we have that
$$L=L^{*}\odot L_{1} \odot L_{2}(\textbf{Q})= \textbf{Q}_{i\notin\mathcal{L}}\cup\{Q^{**}_{\mathcal{L}_{\{1,2\}}}\},$$
where $Q^{**}_{\mathcal{L}_{\{1,2\}}}$ is given by
$$Q^{**}_{\mathcal{L}_{\{1,2\}}}=\dfrac{Q^{*}_{1}e^{w^{*}_{1}}+Q^{*}_{2}e^{w^{*}_{2}}}{e^{w^{*}_{1}}+e^{w^{*}_{2}}}.$$
As such to prove our claim we need to show that
$$\dfrac{Q^{*}_{1}e^{w^{*}_{1}}+Q^{*}_{2}e^{w^{*}_{2}}}{e^{w^{*}_{1}}+e^{w^{*}_{2}}}=\dfrac{\sum_{i\in\mathcal{L}} Q_i e^{w_i }}{\sum_{i\in\mathcal{L}}{e^{w_i}}}.$$
Now note that $\exp(w_{1}^{*})=\sum_{j\in \mathcal{L}_{1}}\exp(w_{j})$. Thus
$$Q^{*}_{\mathcal{L}_{1}}e^{w_{1}^{*}}=\sum_{i \in \mathcal{L}_{1}}Q_{i},$$
and as such we precisely have that 
$$\dfrac{Q^{*}_{1}e^{w^{*}_{1}}+Q^{*}_{2}e^{w^{*}_{2}}}{e^{w^{*}_{1}}+e^{w^{*}_{2}}}=\dfrac{\sum_{i\in\mathcal{L}} Q_i e^{w_i }}{\sum_{i\in\mathcal{L}}{e^{w_i}}},$$
which concludes our proof for $\textbf{Q}$. Similar arguments show the equivalence for $\Xi,\textbf{w}$.
\end{proof}
\par As such we see that our transformations are commutative and the sequential application of transformations corresponds to the application of one transformation on the entire set of indices.
\subsection{Invariance of dynamics to transformation for identical memory patterns}

We will show that the application of the transformation in Eq.~\ref{eqn:coarse_graining_transformations} on identical patterns (identical rows of $\Xi$) leaves the dynamics invariant. This has been previously demonstrated in Karin \cite{karin2024enhancernet} and is included here for completeness. 

\begin{proposition}
Consider a subset $\mathcal{L}$ such that for each $r,s\in\mathcal{L}$ we have that $\Xi_r=\Xi_s$. Then $g(\textbf{x},\beta,\textbf{Q},\Xi,\textbf{w})=g_{\mathcal{L}}(\textbf{x},\beta,\textbf{Q},\Xi,\textbf{w})$ .
\end{proposition}
\begin{proof}
Consider the original dynamics for $x_j$:
\begin{equation*}
 \begin{split}
 \dot{x_j} &= \sum_{i=1}^{K} {q_{i,j} \frac{e^{w_i +\beta \sum_l {\xi_{i,l}x_l} }}{\sum_{k}{e^{w_k +\beta \sum_l {\xi_{k,l}x_l} }}}} - x_j = \\
 &=\sum_{i\in \mathcal{L}} { \frac{q_{i,j}e^{w_i +\beta \sum_l {\xi_{i,l}x_l} }}{\sum_{k\in{\mathcal{L}}}{e^{w_k +\beta \sum_l {\xi_{k,l}x_l} }}+\sum_{k\notin{\mathcal{L}}}{e^{w_k +\beta \sum_l {\xi_{k,l}x_l} }}}} +\sum_{i\notin \mathcal{L}} { \frac{q_{i,j}e^{w_i +\beta \sum_l {\xi_{i,l}x_l} }}{\sum_{k\in{\mathcal{L}}}{e^{w_k +\beta \sum_l {\xi_{k,l}x_l} }}+\sum_{k\notin{\mathcal{L}}}{e^{w_k +\beta \sum_l {\xi_{k,l}x_l} }}}} - x_j 
 \end{split} 
\end{equation*}
Now, under the transformation, and noting that for $k\in\mathcal{L}$ we have that $\Xi_k=\Xi^{*}_{\mathcal{L}}$, we have:
\begin{equation*}
 \sum_{k\in{\mathcal{L}}}{e^{w_k +\beta \sum_l {\xi_{k,l}x_l} }}= \sum_{k\in{\mathcal{L}}}{e^{w_k +\beta \Xi^{*}_{\mathcal{L}} 
 \cdot \textbf{x}}} = e^{w^{*}_{\mathcal{L}} + \beta \Xi^{*}_{\mathcal{L}} 
 \cdot \textbf{x}}
\end{equation*}
And:

\begin{equation*}
 \sum_{k\in{\mathcal{L}}}q_{k,j}{e^{w_k +\beta \sum_l {\xi_{k,l}x_l} }}= e^{\beta \Xi^{*}_{\mathcal{L}}\cdot \textbf{x}}\sum_{k\in{\mathcal{L}}}{q_{k,j} e^{w_k} 
 }=( Q^*_{\mathcal{L}})_j e^{w^{*}_{\mathcal{L}}+\beta \Xi^{*}_{\mathcal{L}}\cdot \textbf{x}} 
\end{equation*}
And thus the transformed dynamics are given by:
\begin{equation*}
\begin{split}
 \dot{x_j}&=\frac{( Q^*_{\mathcal{L}})_j e^{w^{*}_{\mathcal{L}}+\beta \Xi^{*}_{\mathcal{L}}\cdot \textbf{x}} }{e^{w^{*}_{\mathcal{L}} + \beta \Xi^{*}_{\mathcal{L}} 
 \cdot \textbf{x}}+\sum_{k\notin{\mathcal{L}}}{e^{w_k +\beta \sum_l {\xi_{k,l}x_l} }}}+\sum_{i\notin \mathcal{L}} {q_{i,j} \frac{e^{w_i +\beta \sum_l {\xi_{i,l}x_l} }}{e^{w^{*}_{\mathcal{L}} + \beta \Xi^{*}_{\mathcal{L}} 
 \cdot \textbf{x}}+\sum_{k\notin{\mathcal{L}}}{e^{w_k +\beta \sum_l {\xi_{k,l}x_l} }}}} - x_j \\
 &= \sum_{i\notin{\mathcal{L}} \text{ or } i=*} {q_{i,j} \frac{e^{w_i +\beta \sum_l {\xi_{i,l}x_l} }}{\sum_{k}{e^{w_k +\beta \sum_l {\xi_{k,l}x_l} }}}} - x_j
 \end{split} 
\end{equation*}
And thus the dynamics are invariant under the transformation.
\end{proof}

\subsection{Symmetric dynamics derived from autoregulatory dynamics}

A recurring motif, across dozens of cell types, is autoregulation, where specific sets of TFs co-bind adjacent enhancers \cite{saint2016models}. Within our modelling framework, such a motif corresponds to a set $\mathcal{L}$ such that for $i\in\mathcal{L}$, we have that $\Xi_i=\textbf{z}$ where $\textbf{z}$ has large positive entries for all $i\in\mathcal{L}$ and zero entries for $i\notin{\mathcal{L}}$; and that $Q_i$ has a large positive entry at index $i$ and zero otherwise. Coarse-graining on $\mathcal{L}$, which is associated with identical memory patterns, leaves the dynamics invariant; and its associated coarse-grained patterns are $\Xi_{\mathcal{L}}^*=\textbf{z}$ and $Q_{\mathcal{L}}^*\approx\textbf{z}$ since it is averaged of all $Q_i$. Thus, the effective dynamics of the cell identity network approximates the symmetric dynamics of Eq.~\ref{eqn:energy_function}.

\subsection{Coarse-graining for general patterns of identical magnitude}

\subsubsection{Similarity of the coarse-grained energy function in the two-pattern case}
\par Throughout this section we suppose that the patterns in question have equal magnitude, which without loss of generality we set to unity. Our analysis will focus on the symmetric case, namely when $\textbf{Q}=\Xi$ so that we have an energy functional:
\begin{equation}\label{eqn:energy_function}
 V(\textbf{x},\beta,\Xi,\textbf{w})=-\beta^{-1}\log(Z)+||\textbf{x}||^{2}/2
\end{equation}
with $Z=\sum_i{e^{\beta\Xi_i\cdot \textbf{\textbf{x}}+w_i}}$. Notice that the values that $V(\textbf{x})$ takes are, in general, of order unity. The main idea is that now, while the transformed dynamics are not identical to the original dynamics, they are similar to the original dynamics up to a perturbative term of order $\beta \mu^2$ with $\mu$ being the angle between the patterns; and thus for a small enough $\beta$ the dynamics will be effectively identical.
To see this, consider a coarse-graining transformation on the indices $\mathcal{L}$ and let $V_{\mathcal{L}}(\textbf{x})$ be the coarse-grained energy functional. Then we have that 
$$V_{\mathcal{L}}(\textbf{x})-V(\textbf{x})=\dfrac{1}{\beta}\left( \log\left(\sum_{i=1}^{K}e^{(w_{i}+\beta \Xi_{i} \cdot \textbf{x})}\right)-\log\left(e^{\left(w^{*}_{\mathcal{L}}+\beta\Xi^{*}_{\mathcal{L}} \cdot \textbf{x}\right)}+\sum_{i \notin \mathcal{L}}e^{(w_{i}+\beta\Xi_{i} \cdot \textbf{x})}\right) \right).$$
And thus if we let $R=\sum_{i \notin \mathcal{L}}e^{w_{i}+\beta\Xi_{i}\cdot \textbf{x}}$, we have
$$V_{\mathcal{L}}(\textbf{x})-V(\textbf{x})=\dfrac{1}{\beta}\left( \log(R+\sum_{i \in \mathcal{L}}e^{(w_{i}+\beta\Xi_{i} \cdot \textbf{x})})-\log(R+e^{(w^{*}_{\mathcal{L}}+\beta\Xi^{*}_{\mathcal{L}} \cdot \textbf{x})}) \right).$$
If $\mathcal{L}$ is just two indices, $\mathcal{L}=\{r,s\}$, the above reads

$$V_{\mathcal{L}}(\textbf{x})-V(\textbf{x})=\dfrac{1}{\beta}\left( \log(R+[e^{(w_{r}+\beta\Xi_{r} \cdot \textbf{x})}+e^{(w_{s}+\beta\Xi_{s} \cdot \textbf{x})}])-\log \left (R+[e^{w_{s}}+e^{w_{r}}]e^{\beta\dfrac{e^{w_{s}}\Xi_{s}\cdot \textbf{x}+e^{w_{r}}\Xi_{r}\cdot \textbf{x}}{e^{w_{s}}+e^{w_{r}}}}\right) \right).$$

\par Let $\theta-\phi,\theta+\phi,\mu_{r,s}$ denote $\angle ({\Xi_r}, {\textbf{\textbf{x}}}),\angle ({\Xi_s},{\textbf{\textbf{x}}}),\angle ({\Xi_r}, {\Xi_s})$, noting that from the triangle inequality of angles we have that $\phi\leq\mu_{r,s}/2$. Recall that we assume $||\Xi_{r}|| = ||\Xi_{s}|| = 1$. Then we have that $\beta(V_{\mathcal{L}}(\textbf{x})-V(\textbf{x}))$ is equal to
$$ \log(R+[e^{(w_{r}+\beta||\textbf{x}||\cos(\theta-\phi))}+e^{(w_{s}+\beta||\textbf{x}||\cos(\theta+\phi))}])-\log(R+[e^{w_{s}}+e^{w_{r}}]e^{\beta||\textbf{x}||\dfrac{e^{w_{s}} \cos(\theta-\phi)+e^{w_{r}}\cos(\theta+\phi)}{e^{w_{s}}+e^{w_{r}}}}).$$

\par Expanding around a small angle $\phi \ll 1$, we get that 
\begin{equation*}
 V_{\mathcal{L}}(\textbf{x})-V(\textbf{x})=
 \dfrac{2\beta||\textbf{x}||^{2}\sin^{2}(\theta)\phi^{2}e^{\beta||\textbf{x}||\cos(\theta)+w_{r}+w_{s}}}{(e^{w_{r}}+e^{w_{s}})(e^{\beta ||\textbf{x}|| \cos(\theta)}(e^{w_{r}}+e^{w_{s}})+R)}+ \mathcal{O}(\phi^{3}).
\end{equation*}
By letting
$$\nu= \dfrac{2e^{\beta||\textbf{x}||\cos(\theta)+w_{r}+w_{s}}}{(e^{w_{r}}+e^{w_{s}})(e^{\beta||\textbf{x}||\cos(\theta)}(e^{w_{r}}+e^{w_{s}})+R)}, $$
we can find an upper bound in the following way.
$$\nu \leq \dfrac{2e^{w_{r}+w_{s}}}{(e^{w_{r}}+e^{w_{s}})^{2}} \leq \dfrac{1}{2}, $$
with the final inequality following from the AM-GM inequality. From this we in particular have that
$$ V_{\mathcal{L}}(\textbf{x})-V(\textbf{x})= 
 \beta\nu||\textbf{x}||^{2}\sin^{2}(\theta)\phi^{2}+ \mathcal{O}(\phi^{3}) \leq \frac{1}{8} \beta\mu^{2}$$
where in the last inequality we also used $||\textbf{x}||\leq 1$.

\subsubsection{Coarse-graining for subsets of many patterns}
Our previous results can be generalized in a straightforward manner to subsets of many patterns. Consider a set of unit vector patterns $\Xi_{1},\dots,\Xi_{n}$ and a vector $\textbf{x}$ such that the angle between $\textbf{x}$ and pattern $i$ is denoted $\angle (\Xi_i,\textbf{x})=\theta+k_i\phi$ (we may set $\theta$ as the minimal angle between $\textbf{x}$ and any of the patterns). We denote the set of indices of patterns we are coarse-graining by $\mathcal{L}$, and we set $\angle (\Xi_i,\Xi_j)=\mu_{ij}$ (with $\mu_{\mathcal{L}}$ as the maximal angle between any two patterns in $\mathcal{L}$). We have, from the triangle inequality of angles:
\begin{equation*}
 \mu_{ij}+\theta + k_i \phi \geq \theta + k_j \phi
\end{equation*}
And thus:
\begin{equation*}
 \mu_{ij}\geq (k_j-k_i) \phi
\end{equation*}
Consider the energy difference:
\begin{equation*}
 \begin{split}
 V_{\mathcal{L}}(\textbf{x})-V(\textbf{x})&= 
\dfrac{1}{\beta}\left( \log(R+\sum_{i\in\mathcal{L}}{e^{w_i + \beta\Xi_{i}\cdot\textbf{x}}}) -
\log \left (R+ e^{ w^*_{\mathcal{L}} + \beta \Xi^*_{\mathcal{L}}\cdot \textbf{x}}\right) \right)\\
&=\dfrac{1}{\beta}\left( \log(R+\sum_{i\in\mathcal{L}}{e^{w_i + \beta\Xi_{i}\cdot\textbf{x}}})-
\log \left (R+\sum_{i\in\mathcal{L}}{e^{w_i}} e^{\beta\dfrac{\sum_{i\in{\mathcal{L}}}e^{w_i}\Xi_i \cdot \textbf{x}}{\sum_{i\in\mathcal{L}}{e^{w_i}}}}\right) \right) \\
&=\dfrac{1}{\beta}\left( \log(R+\sum_{i\in\mathcal{L}}{e^{w_i + \beta ||\textbf{x}||\cos{\left(\theta+k_i\phi\right)}}})-
\log \left (R+\sum_{i\in\mathcal{L}}{e^{w_i}} e^{\dfrac{\sum_{i\in{\mathcal{L}}}e^{w_i} \beta ||\textbf{x}||\cos{\left(\theta+k_i\phi\right)}}{\sum_{i\in\mathcal{L}}{e^{w_i}}}}\right) \right)
 \end{split}
\end{equation*}
We can expand over $\phi$ to a second order:
\begin{equation*}
\begin{split}
 |V_{\mathcal{L}}(\textbf{x})-V(\textbf{x})| &\approx \frac{\sum _{i,j \in \mathcal{L},i<j}\left(k_i-k_j\right)^2 e^{w_i+w_j}}{\frac{2 \sum _{i\in \mathcal{L}} e^{w_i}}{e^{\beta \cos (\theta )}}R+2 \left(\sum _{i\in \mathcal{L}} e^{w_i}\right)^2}\beta ||\textbf{x}||^2 \phi ^2\sin ^2(\theta ) \\
 &\leq \frac{\sum_{i,j \in \mathcal{L}, i< j }\left(k_i-k_j\right)^2 e^{w_i+w_j}}{\left(\sum_{i\in \mathcal{L}} e^{w_i}\right)^2} \frac{\beta ||\textbf{x}||^2 \phi ^2\sin ^2(\theta )}{2} \\
 &\leq
 \frac{\sum_{i,j \in \mathcal{L}, i< j }{e^{w_i+w_j}}}{\left(\sum_{i\in \mathcal{L}} e^{w_i}\right)^2} \frac{ \beta ||\textbf{x}||^2 \mu_{\mathcal{L}} ^2\sin ^2(\theta )}{2} \\
 &=\frac{\sum_{i,j \in \mathcal{L}, i< j }{e^{w_i+w_j}}}{2\sum_{i,j \in \mathcal{L}, i< j }{e^{w_i+w_j}}+\sum_{i \in \mathcal{L}}{e^{2w_i}}} \frac{ \beta ||\textbf{x}||^2 \mu_{\mathcal{L}} ^2\sin ^2(\theta )}{2} \\ 
 &\leq \frac{\sin ^2(\theta )}{4} \beta \mu_{\mathcal{L}} ^2 \leq \frac{1}{4} \beta \mu_{\mathcal{L}} ^2
\end{split}
\end{equation*}
The coarse grained dynamics are thus similar to the fine-grained dynamics when $\beta \ll 1/\mu_{\mathcal{L}} ^2$.

\subsection{Coarse-graining transformation for patterns of arbitrary magnitude}

We will now generalize the coarse-graining transformation to be valid for patterns of any magnitude. The idea behind the generalized transformation is identical to the previous sections - we will look for a transformation for which the transformed dynamics are identical to the original dynamics up to a factor of magnitude $\beta \mu_{\mathcal{L}} ^2$. For a subset of indices $\mathcal{L}\subseteq\{1\dots K\}$ corresponding to patterns of magnitudes $u_i, i\in\mathcal{L}$, we define our coarse-graining transformation by $\textbf{Q}\xrightarrow{\mathcal{L}}\textbf{Q}',\Xi\xrightarrow{\mathcal{L}}\Xi',\textbf{w}\xrightarrow{\mathcal{L}} \textbf{w}'$, where:
\begin{equation*}
 \textbf{Q}'=\textbf{Q}_{i\notin\mathcal{L}}\cup\{Q^*_{\mathcal{L}}\}, \;\Xi'=\Xi_{i\notin\mathcal{L}}\cup\{\Xi^*_{\mathcal{L}}\},\;\textbf{w}'=\textbf{w}_{i\notin\mathcal{L}}\cup\{w^*_{\mathcal{L}}\}
\end{equation*}
with:
\begin{equation}
\begin{split}
 Q^*_{\mathcal{L}}&= \frac{\sum_{i\in\mathcal{L}} Q_i e^{w_i+\beta u_i }}{\sum_{i\in\mathcal{L}}{e^{w_i+\beta u_i}}} = \text{softmax}(\textbf{w}_{i\in\mathcal{L}}+\beta\textbf{u}_{i\in\mathcal{L}}) Q_{i\in\mathcal{L}}
 \\
 \Xi^*_{\mathcal{L}}&= \frac{\sum_{i\in\mathcal{L}} \Xi_i e^{w_i+\beta u_i }}{\sum_{i\in\mathcal{L}}{e^{w_i+\beta u_i}}} = \text{softmax}(\textbf{w}_{i\in\mathcal{L}}+\beta\textbf{u}_{i\in\mathcal{L}}) \Xi_{i\in\mathcal{L}}\\
 w^*_{\mathcal{L}}&=-\frac{\sum_{i\in\mathcal{L}}{u_i e^{w_i + u_i \beta}}}{\sum_{i\in\mathcal{L}}{e^{w_i + u_i \beta}}}\beta + \log{\sum_{k\in\mathcal{L}}{e^{w_k+u_k \beta}}} = -\text{softmax}(\textbf{w}_{i\in\mathcal{L}}+\beta\textbf{u}_{i\in\mathcal{L}}) \cdot \beta \textbf{ u} + \log{\sum_{k\in\mathcal{L}}{e^{w_k+u_k \beta}}}
\end{split}
\label{eq:uneqtrans}
\end{equation}
While the transformation appears similar to the equal magnitude case, a key difference is the $\beta$-dependence, where at large $\beta$ the largest pattern dominates the transformation. For an example of the application of the transformation, see Fig.~\ref{fig:unequal_patterns}.

Similar to the equal-magnitude case, this transformation is self-similar. It also maintains symmetry - if $Q_i=\Xi_i$ for all $i\in\mathcal{L}$ then $Q^*_{\mathcal{L}}=\Xi^*_{\mathcal{L}}$. The transformation coarse-grains the energy function in a similar manner to the equal-magnitude case, namely by:
 \begin{equation*}
 V_{\mathcal{L}}(\textbf{x},\beta,\Xi,\textbf{w})= V(\textbf{x},\beta,\Xi',\textbf{w'}) + \mathcal{O}(\beta \mu_{\mathcal{L}}^2)
 \end{equation*}
 where here the $\mathcal{O}(\beta \mu_{\mathcal{L}}^2)$ involves a more complex dependence on the magnitudes.

\begin{figure}[H]
 \centering
 \includegraphics[width=0.8\linewidth]{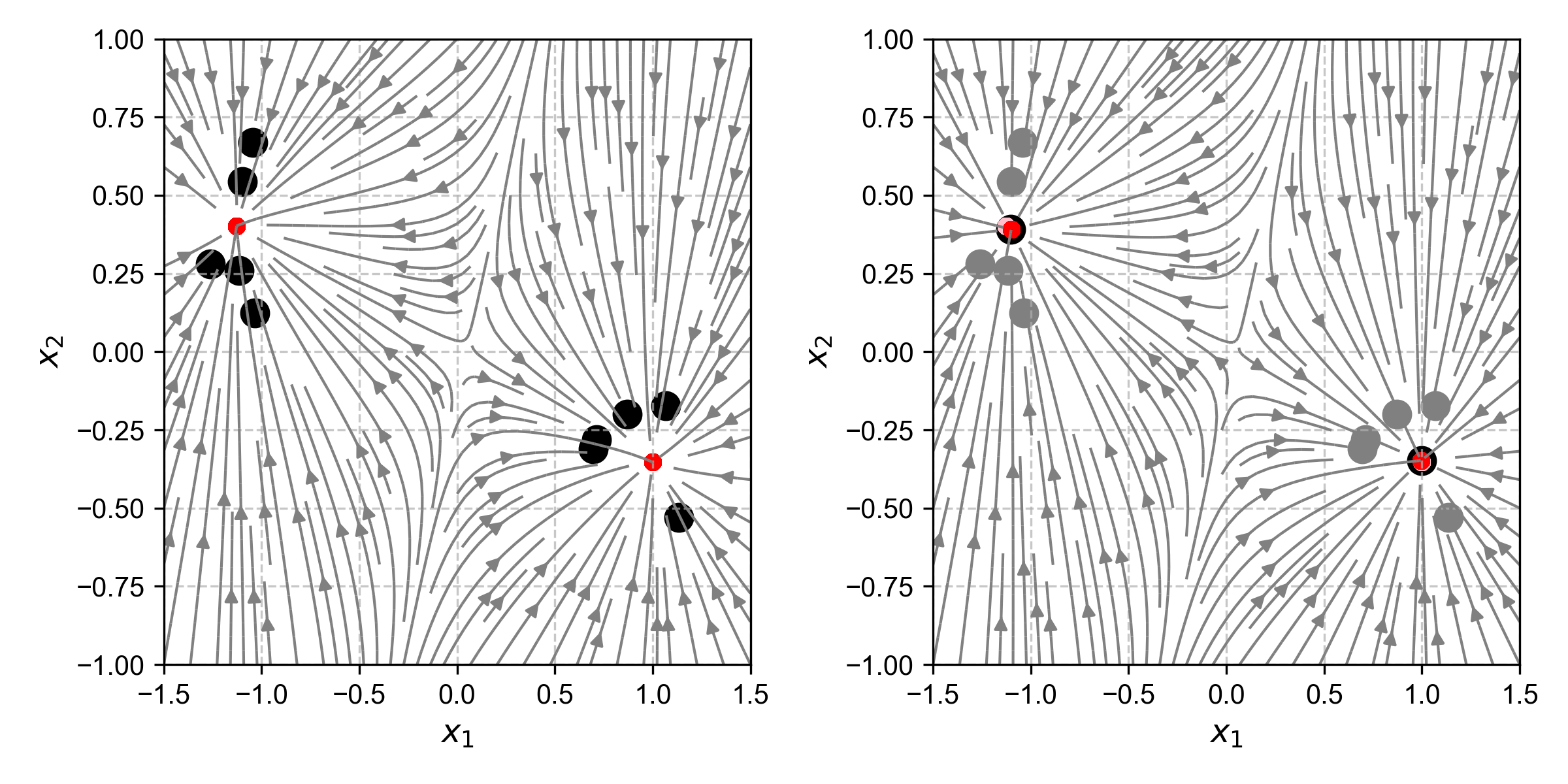}
 \caption{\textbf{Example of application of transformation for patterns of arbitrary magnitude.}} Left: Dynamics at $\beta=3$ for two clustered sets of patterns (black dots correspond to patterns, red dots correspond to stable fixed points). Right: Coarse-grained dynamics with the transformation~\ref{eq:uneqtrans} applied to pattern clusters. Black dots correspond to coarse-grained patterns, red dots correspond to stable fixed points, gray dots correspond to fine-grained patterns and pink dots to stable fixed points of the original dynamics.
 \label{fig:unequal_patterns}
\end{figure}

\section{Dimensionality reduction of the dynamics}

In this section we will demonstrate that it is possible to capture the dynamics of a system in a reduced form based on dimensionality reduction and rotation. Consider the dynamics given by Eq.~\ref{eq:enhancernet}, where the dimensions of $\textbf{Q},\Xi$ are $K\times N$. We can see that $\textbf{x}$ is confined to the convex hull whose vertices are the rows of $\textbf{Q}$ (since $\text{softmax}$ produces a probability vector). This convex hull is a subspace of dimension at most $K-1$. In addition, the dynamics of Eq.~\ref{eq:enhancernet} are invariant to orthogonal transformations. This is the subject of the following proposition:
\begin{proposition} \label{prop: orthogonal_transform}
 Let $\mathbf{M} \in \mathbb{R}^{N \times N}$ be an orthogonal matrix. Then if $\mathbf{x}$ satisfies Eq.~\ref{eq:enhancernet}, with $\mathbf{Q},\Xi$, then $\mathbf{y}=\mathbf{Mx}$, satisfies Eq.~\ref{eq:enhancernet} with $\tilde{\mathbf{Q}}=\mathbf{Q}\mathbf{M}^{T},\tilde{\Xi}=\Xi \mathbf{M}^{T}$.
\end{proposition}
\begin{proof}
 We have that $\dot{\mathbf{y}}=\mathbf{M}\dot{\mathbf{x}}$ and so from Eq.~\ref{eq:enhancernet} 
 $$\dot{\mathbf{y}}=\mathbf{M}\mathbf{Q}^{T}\text{softmax}\left(\beta \Xi \textbf{x} + \textbf{w} \right) - \textbf{y}$$
By the fact that $\mathbf{M}$ is an orthogonal matrix we have $\mathbf{M}^{T}\mathbf{M}=I$ and so 
 $$\dot{\mathbf{y}}=\tilde{\mathbf{Q}}^{T}\text{softmax}\left(\beta \tilde{\Xi} \textbf{y} + \textbf{w} \right) - \textbf{y}.$$
\end{proof}

\par Now we will examine how the invariance of the dynamics of Eq.~\ref{eq:enhancernet} to orthogonal transformations, can lead to the decoupling of the dynamics of some of the components, and how in the symmetric case, where we have a potential given by Eq.~\ref{eqn:energy_function}, it leads to a dimensional reduction. Consider the case where $N > K$, and let $\mathbf{A}$ be a $K\times N$ real matrix. We may apply the singular value decomposition to get $\mathbf{A}=\mathbf{U}\mathbf{S}\mathbf{V}^{T},$ where $\mathbf{U}$ is a $K \times K$ orthogonal matrix, $\mathbf{V}$ is a $N \times N$ orthogonal matrix, and $\mathbf{S}$ is a $K \times N$ matrix, with the diagonal entries corresponding to the singular values of $\mathbf{A}$, and the rest of the entries being zeros. Crucially note that \emph{the number of non-zero entries of $\mathbf{S}$ is equal to the rank of $\mathbf{A}$}. Multiplying by $\mathbf{V}$, and calling this quantity, $\tilde{\mathbf{A}}$, we get that
$$\tilde{\mathbf{A}}:= \mathbf{A}\mathbf{V}=\mathbf{U}\mathbf{S}.$$
Now consider the form $\tilde{\mathbf{A}}$ takes. By choosing appropriate $\mathbf{V}, \mathbf{U}$ in the decomposition, we may always ensure that the diagonal entries of $\mathbf{S}$ appear in descending order, i.e. $d_{1} \geq d_{2} \geq d_{K} \geq 0$ (note that singular values will always be non-negative). I.e. $\mathbf{S}$ has the following block form

\[
\mathbf{S} = \begin{pNiceArray}{ccccc|c}
 d_{1} & 0 &0 & ... & 0 & \Block{5-1}<\Large>{\mathbf{0}} \\
 0 & d_{2} & 0 &... &0 \\
 \Vdots & & \Ddots & & \Vdots \\
 0& ... &0 & d_{K-1} & 0\\
 0 &\Cdots & & 0 & d_{K}
\end{pNiceArray}
\]
So if we look at the $ij^{\text{th}}$ component of $\tilde{\mathbf{A}}$, we get 
$$\tilde{a}_{ij}=\sum_{r=1}^{K}u_{ir}s_{rj}=\begin{cases}
 \sum_{r=1}^{K}u_{ir}\delta_{rj}d_{j} \text{ if } j \leq K\\
 0 \text{ otherwise }
\end{cases}.$$
And thus
$$\tilde{a}_{ij}=\begin{cases}
 u_{ij}d_{j} \text{ if } j \leq K\\
 0 \text{ otherwise }
\end{cases}.$$
Thus for any column $m$, of $\tilde{\mathbf{A}}$ such that $N\geq m >K$, all the entries of that column will be 0. 

Therefore returning to the dynamics of Eq.~\ref{eq:enhancernet}, and considering the singular value decomposition for $\mathbf{A}=\mathbf{Q}$, with the transformation of proposition \ref{prop: orthogonal_transform} given by $\mathbf{M}^{T}=\mathbf{V}$, we get that the dynamics of our equation will be given by 
$$\dot{\mathbf{y}}=\tilde{\mathbf{Q}}^{T}\text{softmax}\left(\beta \tilde{\Xi} \textbf{y} + \textbf{w} \right) - \textbf{y},$$
and from our particular choice of the transformation we know that for $K<j \leq N$, we have
\begin{equation} \label{eqn:decoupled components dynamics}
 \dot{y_{j}}=-y_{j},
\end{equation}
since 
$$\sum_{k=1}q_{kj}[\text{softmax}(\beta\tilde{\Xi}\mathbf{y}-\mathbf{w})]_{k}=0.$$

\par From this we deduce that our system decouples, with the dynamics of $y_{j}$ for $j \geq K$ being completely determined\footnote{In fact since the number of non-zero singular values is the same as the rank of the matrix, if $\mathbf{Q}$ has a smaller rank, $R<K$ then $y_{j}$ will decouple for all $j >R$ }. Now this means that in the general case of the dynamics being given by Eq.~\ref{eq:enhancernet}, provided our dynamics start at $\mathbf{0}$, we can reduce the dynamics of our system from $N$ to $K$ dimensions.

However, in the symmetric case when $\mathbf{Q}=\Xi$, i.e. our system is given by a potential function as in Eq.~\ref{eqn:energy_function}, we do get a full decoupling of the system, as now 
$$\sum_{r=1}^{N}\tilde{\xi}_{ir}y_{r}=\sum_{r=1}^{K}\tilde{\xi}_{ir}y_{r},$$
as $\mathbf{Q}=\Xi$ and so $\tilde{\xi}_{ir}=0$, for $r>K$. And as such, with our orthogonal transformation, we reduce the dimension of the system from $N$ to $K$.

\section{Dynamics in the vicinity of the coarse-grained pattern} \label{section:Dynamcis_vicinity}
\par We want to understand the local dynamics of the system. One would expect that when you are close to a coarse-grained pattern $\Xi^{*}$, the patterns that contribute to the dynamics, are precisely the ones that form the coarse-grained pattern. This is indeed the case as we will see below, for a parameter regime, which we will derive in this section.

\par As before suppose all patterns have a magnitude of unity and suppose we have two isolated patterns in our system, $\Xi_{r},\Xi_{s}$, i.e. we have that for $l \neq r,s$, $\mu_{r,s} \ll \mu_{l,r}$ and $\mu_{r,s} \ll \mu_{l,s}$, where $\mu_{ij}$ is the angle between patterns $i$ and $j$. Let us also define $\mu_{m}$ to be the value, such that $\mu_{m}=\min_{l\neq r,s}\{\mu_{l,s},\mu_{l,r}\}$, i.e. $\mu_{m}$ is the smallest angle of separation between the isolated patterns and the rest of the patterns. We are interested in the dynamics in the vicinity of the progenitor pattern of the two isolated patterns, that is, we consider $\textbf{x}$ such that 
$$\textbf{x}=\Xi^{*}+\textbf{y},$$
where $\textbf{y}$ is some perturbation away from the coarse-grained pattern. The only restriction we place on $\textbf{y}$ is that it must lie in the region "between" the two patterns $\Xi_{r},\Xi_{s}$, i.e. for any non-isolated pattern $\Xi_{i}$, we must have $\angle(\Xi_{i},\textbf{y})\geq \mu_{m}$, while for the isolated patterns we have $\angle(\textbf{y},\Xi_{r}) \leq \angle(\Xi_{r},\Xi_{s})$ and $\angle(\textbf{y},\Xi_{s}) \leq \angle(\Xi_{r},\Xi_{s})$. Then our dynamical system can be written as
$$\dot{\textbf{y}} = \textbf{Q}^T \text{softmax}\left(\beta \Xi (\Xi^{*}+\textbf{y}) + \textbf{w} \right) - \textbf{y}-\Xi^{*}.$$
Writing out the $i^{\text{th}}$ component of the soft-max, we have 
\begin{align} 
\left[\text{softmax}\left(\beta \Xi (\Xi^{*}+\textbf{y})+\textbf{w}\right)\right]_{i} &=\dfrac{1}{Z}e^{\left(\beta (\Xi_{i} \cdot \Xi^{*}+\Xi_{i} \cdot \textbf{y})+w_{i}\right)} \\ \label{eqn:softmax_gen_component}
&=\dfrac{1}{Z}e^{\beta \left(\dfrac{\cos(\mu_{i,r})e^{w_{r}}+\cos(\mu_{i,s})e^{w_{s}}}{e^{w_{r}}+e^{w_{s}}}+\Xi_{i} \cdot \textbf{y}\right)+w_{i}}
\end{align}
 Now consider in particular a pattern, $\Xi_{i}$, which is not one of the two isolated ones, so that $i \neq r,s$. Now $\mu_{m}\leq \mu_{r,s}$ and as such $\cos(\mu_{m}) \geq\cos(\mu_{r,s})$. Therefore we can bound the $i^{\text{th}}$ soft-max component, as follows
 \begin{equation*}
\left[\text{softmax}\left(\beta \Xi (\Xi^{*}+\textbf{y})+\textbf{w}\right)\right]_{i} \leq \dfrac{1}{Z}e^{\beta(\cos(\mu_{m})+\Xi_{i} \cdot \textbf{y})+w_{i}}
\end{equation*}
Now recall we consider $\textbf{y}$, such that the angle between $\textbf{y}$ satisfies $\angle(\Xi_{i},\textbf{y})\geq \mu_{m}$, and therefore $\cos(\angle(\Xi_{i},\textbf{y})) \leq \cos(\mu_{m})$. From this we get the further bound
 \begin{equation} \label{eqn:softmax_non_isolated}
\left[\text{softmax}\left(\beta \Xi (\Xi^{*}+\textbf{y})+w\right)\right]_{i} \leq \dfrac{1}{Z}e^{\beta\cos(\mu_{m})(1+||\textbf{y}||)+w_{i}}
\end{equation}
\par If we consider (without loss of generality) the $r^{\text{th}}$ component of the soft-max (i.e. one corresponding to one of the isolated patterns), we get from Eq.~\ref{eqn:softmax_gen_component} that
\begin{equation*}
 \left[\text{softmax}\left(\beta \Xi (\Xi^{*}+\textbf{y})\right)\right]_{r}=\dfrac{1}{Z}e^{\beta\left(\dfrac{e^{w_{r}}+\cos(\mu_{r,s})e^{w_{s}}}{e^{w_{r}}+e^{w_{s}}}+\Xi_{r} \cdot \textbf{y}\right)+w_{i}}.
\end{equation*}
Recall that we chose $\textbf{y}$ above, such that we have $\angle(\textbf{y},\Xi_{r}) \leq \angle(\Xi_{r},\Xi_{s})$ and as such $\cos(\angle(\Xi_{r},\Xi_{s})) \geq \cos(\mu_{r,s})$. By further setting $p=e^{w_{r}}/(e^{w_{r}}+e^{w_{s}})$, we get 

\begin{equation} \label{eqn:softmax_isolated}
 \left[\text{softmax}\left(\beta \Xi (\Xi^{*}+\textbf{y})\right)\right]_{r} \geq\dfrac{1}{Z}e^{\beta (p+(1-p)\cos(\mu_{r,s})+||\textbf{y}||\cos(\mu_{r,s}))+w_{i}}
\end{equation}

\par With this preliminary calculation, we now look at the ratio of the soft-max components of a non-isolated pattern, compared to an isolated pattern. Let's call this ratio $\mathcal{S}$, i.e.
\begin{equation*}
 \mathcal{S}= \dfrac{\left[\text{softmax}\left(\beta \Xi (\Xi^{*}+\textbf{y})+w\right)\right]_{i}}{\left[\text{softmax}\left(\beta \Xi (\Xi^{*}+\textbf{y})+w\right)\right]_{r}}. 
\end{equation*}
The reason we are looking at this ratio is as follows. If $\mathcal{S} \ll 1$, then that would mean that the isolated pattern dominates the non-isolated one, and as such we are safe to ignore the contribution of the non-isolated patterns. While if $ \mathcal{S} \cancel{\ll} 1 $, then the non-isolated patterns do contribute to an extent, which we cannot ignore. With this we look at $\mathcal{S}$, and by combining Eq.~\ref{eqn:softmax_non_isolated} and Eq.~\ref{eqn:softmax_isolated} we get
\begin{equation*}
 \mathcal{S} \leq \dfrac{e^{\beta\cos(\mu_{m})(1+||\textbf{y}||)+w_{i}}}{ e^{\beta (p+(1-p)\cos(\mu_{r,s})+||\textbf{y}||\cos(\mu_{r,s}))+w_{i}}}
\end{equation*}
And thus simplifying we have
\begin{equation}\label{eqn:s_r_big_bound}
 \mathcal{S} \leq e^{\beta[\{\cos(\mu_{m})-\cos(\mu_{r,s})\}(1+||\textbf{y}||)-p+p\cos(\mu_{r,s})]+(w_{i}-w_{r})}.
\end{equation}
Note when $\log(\mathcal{S})\ll 0$, we in particular have $\mathcal{S} \ll 1$. We will look at the logarithmic inequality corresponding to Eq.~\ref{eqn:s_r_big_bound}, 
$$\log(\mathcal{S}) \leq \beta\{\cos(\mu_{m})-\cos(\mu_{r,s})\}(1+||\textbf{y}||)-p+p\cos(\mu_{r,s})+w_{i}-w_{r}.$$
To have a more explicit picture, let us consider an approximation for the angle between the isolated patterns, $\mu_{r,s}$, and for $\mu_m$. Proceeding via a Taylor expansion, we have
\begin{equation*}
 \begin{split}
 \log(\mathcal{S}) &\leq \beta\left(\dfrac{\mu_{r,s}^{2}}{2}-\dfrac{\mu_{m}^{2}}{2}\right)(1+||\textbf{y}||)-\dfrac{p\mu_{r,s}^{2}}{2}+w_{i}-w_{r}+\mathcal{O} (\mu_{r,s}^{4},\mu_{m}^{4}) \\
 &\leq \beta\left(\dfrac{\mu_{r,s}^{2}}{2}-\dfrac{\mu_{m}^{2}}{2}\right)(1+||\textbf{y}||)+w_{i}-w_{r}+\mathcal{O} (\mu_{r,s}^{4},\mu_{m}^{4}) 
 \end{split}
\end{equation*}
Now, recall that we assumed that $\mu_{r,s}\ll \mu_m$. As such, for sufficiently large $\beta$ (when $\beta \mu_{r,s}^2$ is of order unity, compared with $w_i,w_r$ that are of order unity), we will expect to have $ \log(\mathcal{S})\ll 0$ and the non-isolated patterns do not contribute to the dynamics in the vicinity of $\Xi^{*}$.

\section{$\textbf{w}$-Modulation of the energy function and basins of attraction}
In this section we explore how $\textbf{w}$ changes the energy landscape of our system. We aim to show that an increase in $w_{i}$ leads to a decrease of the energy in the neighbourhood of $\Xi_{i}$, while leaving the energy around the other patterns virtually unchanged. As such we consider the symmetric case of $\textbf{Q}=\Xi$, so that we have an energy function $V(\textbf{x})$ given by Eq.~\ref{eqn:energy_function}. We focus on how the modulation of one component of $\textbf{w}$, $w_{i}$, affects the energy landscape in the vicinity of our stable patterns.

\par To do this let us consider, how a small increase $\delta$ in $w_{i}$, affects the energy landscape. Let $\tilde{V}(\textbf{x})$ be the modified energy landscape, i.e. one where $\tilde{w_{i}}=w_{i}+\delta$, then 
$$\tilde{V}(\textbf{x})-V(\textbf{x})=\dfrac{1}{\beta}(\log(Z)-\log(\tilde{Z}))=\dfrac{1}{\beta}\left(\log\left(\dfrac{R+e^{\beta\Xi_{i}\cdot \textbf{x}+w_{i}}}{R+e^{\beta\Xi_{i}\cdot \textbf{x}+w_{i}+\delta}}\right)\right),$$
where $R=\sum_{j\neq i} e^{\beta \Xi_{j}\cdot \textbf{x}+w_{j}}$. Considering the case where $\delta$ is small and doing a Taylor expansion we get
$$\tilde{V}(\textbf{x})-V(\textbf{x})= -\dfrac{e^{\beta\Xi_{i}\cdot \textbf{x}+w_{i}}}{R+e^{\beta\Xi_{i}\cdot \textbf{x}+w_{i}}}\delta+ \mathcal{O}(\delta^{2})= -s_{i}\delta +\mathcal{O}(\delta^{2}),$$
where $s_{i}$ is the $i^{\text{th}}$ component of the soft-max function. From this form, it is easy to see what a small increase in $w_{i}$ will do to the energy landscape. Namely in the vicinity of $\Xi_{i}$, where $1 \approx s_{i}\gg s_{j}$, for all $j \neq i$, we get a significant decrease in the energy, while close to other patterns, where $s_{i} \approx 0$, the change in the energy will be negligible.

\subsection{Estimation of the basins of attractions}
\par In this section, we aim to derive a heuristic estimate for the basins of attractions. We aim to find the equations for the separatrices. Once again we assume that the dynamics arise from a potential function $V(\textbf{x})$ given by Eqn.~\ref{eqn:energy_function}, so that $\textbf{Q}=\Xi$.
\subsubsection{Theoretical considerations}

\par Assume we have a fixed point, $p$ close to one of our patterns, $\Xi_{i}$ say, i.e. it's the fixed point associated with $\Xi_{i}$.
Suppose then we have a closed and convex region $\Omega \subseteq \mathbb{R}^{N}$,
around $p$, and for all $\mathbf{x} \in \Omega$, $D^{2}V(\textbf{x})>0$. Here $D^{2}V(\textbf{x})$ is the Hessian of $V$ at a point $\textbf{x}$ and by $D^{2}V(\textbf{x})>0$ we mean that the Hessian is positive definite. Then the restriction of our potential to $\Omega$, $V: \Omega \rightarrow \mathbb{R}$ is a strictly convex function. Therefore $V$ has at most 1 minimum in $\Omega$. So since our fixed point $p$ is in $\Omega$, $p$, will be a stable fixed point (as $D^{2}V(p)>0$ by the assumption on $\Omega$) and by the strict convexity it will be the only stable fixed point in $\Omega$. 
\par Now let us define the local sublevel sets of $V$ on $\Omega$ to be $\Lambda(h)$, i.e. $\Lambda(h):=\{\textbf{x}\in \Omega| V(\textbf{x}) \leq h\}$. Note that this set is closed as $V$ is continuous. Also considering $V$ as a function from $\Omega$ to $\mathbb{R}$, which is convex by definition of $\Omega$, we get that $\Lambda_{h}$ is convex, as sublevel sets of a convex function are convex. Therefore as in our dynamical system, $V$ is decreasing ($\dot{V} \leq 0$), $\Lambda_{h}$ is an invariant-set, provided we chose $h$ so that $\Lambda_{h}$ is non-empty and $\Lambda_{h} \cap \partial \Omega= \emptyset$, where $\partial \Omega$ is the boundary of $\Omega$. This implies that for $\Lambda_{h}$ satisfying these properties, $\Lambda_{h}$ is a subset of the basin of attraction of $p$, i.e. all trajectories starting in $\Lambda_{h}$ will converge to $p$.
Thus it will be imperative for us to be able to identify, when $D^{2}V>0$.
\subsubsection{Eigenvalues of the Hessian}

\par A matrix is positive-definite if and only if all it's eigenvalues are real and positive. As such our first task will be to compute the Hessian and figure out its eigenvalues. Note that we already know (from the dynamics of the system), that
$$\dfrac{\partial V(\textbf{x})}{\partial x_{i}}=-(\Xi^{T}[\text{softmax}(\textbf{w}+\beta\Xi\textbf{x})])_{i}+x_{i},$$
and as such by differentiating w.r.t $x_{j}$, we have
$$\dfrac{\partial^{2} V(\textbf{x})}{\partial x_{i} \partial x_{j}}=-\sum_{k=1}^{K} \xi_{ki}[\dfrac{\partial}{\partial x_{j}}\text{softmax}(\textbf{w}+\beta\Xi\textbf{x})]_{k}+\delta_{ij}=-\sum_{k=1}^{K} \xi_{ki}[\dfrac{\partial}{\partial x_{j}}\dfrac{e^{w_{k}+\beta \sum_{m=1}^{N}\xi_{km}x_{m}}}{Z}]+\delta_{ij},$$
where $\delta_{ij}$ is the Kronecker delta and $Z$ is our partition function as before. Thus
\begin{align} 
\dfrac{\partial^{2} V(\textbf{x})}{\partial x_{i} \partial x_{j}}&= -\sum_{k=1}^{K} \xi_{ki}[-\dfrac{e^{w_{k}+\beta \sum_{m=1}^{N}\xi_{km}x_{m}}}{Z^{2}} \dfrac{\partial Z}{\partial x_{j}}+\dfrac{e^{w_{k}+\beta \sum_{m=1}^{N}\xi_{km}x_{m}}}{Z}\beta\xi_{kj}]+\delta_{ij} \\ \label{eqn:Hessian_working_1} &=-\sum_{k=1}^{K} \xi_{ki}[-\dfrac{s_{k}}{Z} \dfrac{\partial Z}{\partial x_{j}}+s_{k}\beta\xi_{kj}]+\delta_{ij}.
\end{align}
where $s_{k}$ is the $k^{\text{th}}$ component of our softmax function.
Looking at the partition function we have
$$\dfrac{\partial Z}{\partial x_{j}}= \sum_{r=1}^{K} \dfrac{\partial}{\partial x_{j}}[e^{w_{r}+\beta\sum_{m}^{N}\xi_{rm}x_{m}}]=\sum_{r=1}^{K} [e^{w_{r}+\beta\sum_{m}^{N}\xi_{rm}x_{m}}]\beta\xi_{rj}.$$
Thus putting this back into Eq.~\ref{eqn:Hessian_working_1}, we have
$$\dfrac{\partial^{2} V(\textbf{x})}{\partial x_{i} \partial x_{j}} = -\sum_{k=1}^{K} \xi_{ki}[-s_{k}\beta \sum_{r=1}^{K}s_{r}\xi_{rj}+s_{k}\beta\xi_{kj}]+\delta_{ij}= -\beta[\sum_{k=1}^{K} \xi_{ki}s_{k}\xi_{kj}-\sum_{k=1}^{K} \sum_{r=1}^{K}\xi_{ki}s_{k}s_{r}\xi_{rj}]+\delta_{ij}. $$
The first term inside the square bracket corresponds to the $ij^{\text{th}}$ component of $\Xi^{T}D_{s}\Xi$, where $D_{s}$ is the diagonal matrix formed from the soft-max vector. While the second term corresponds to the $ij^{\text{th}}$ entry of $\Xi^{T}(s \otimes s)\Xi$, where $s \otimes s$ refers to the outer product of the softmax with itself. Putting it all together, we thus have that 
\begin{equation}\label{eqn:hessian_full}
 D^{2}V(\textbf{x})=I-\beta\Xi^{T}(D_{s}-s \otimes s)\Xi=I-\beta\Xi^{T}J_{s}(\textbf{w}+\beta\Xi\textbf{x})\Xi,
\end{equation}
 where $J_{s}$ is the Jacobian of the softmax, i.e. $[J_{s}]_{ij}=\partial/\partial z_{j}[\text{softmax}(\textbf{z})]$, which we evaluate at $\textbf{z}=\textbf{w}+\beta\Xi\textbf{x}$. (Note that $J_{s}=D_{s}-s \otimes s$), and $s$ is our softmax vector $s=\text{softmax}(\textbf{w}+\beta \Xi \textbf{x})$.

\par Note that as our potential $V(\mathbf{x})$ is a $\mathcal{C}^{2}$ function, it's Hessian matrix is symmetric, and as such all eigenvalues are real. Recall that if we have a positive-definite Hessian ($D^{2}V>0$) at a steady state, the steady state will be stable (as $\dot{\textbf{x}}=-\nabla V(\textbf{x})$) and if the Hessian is indefinite at a steady state, the steady state will be a saddle-point. 

\par Now an equivalent condition for a symmetric matrix to be positive definite is that the minimal eigenvalue, $\lambda_{\min}$ of the matrix is positive, i.e. $\lambda_{\min}>0$. Now in general we cannot find the $\lambda_{\min}$ explicitly for our Hessian, however we can find a subset of the regions, where the Hessian is positive-definite. Namely we have
\begin{proposition} \label{prop:inner_eigen_bound}
If $\lambda_{\max}(J_{s})[\textbf{x}]<||\Xi||^{-2}_{op}\beta^{-1}$ then $D^{2}V(x)>0$, where $||.||_{op}$ is the operator norm\footnote{The operator norm of a linear operator, $A:V \rightarrow W$, where $V,W$ are normed vector spaces, is defined as $||A||_{op}=\inf \{ c: ||Av||_{W}\leq c||v||_{V}\}$} and $\lambda_{\max}(J_{s})[\textbf{x}]$, refers to the maximal eigenvalue of the Jacobian of the softmax evaluated at $\textbf{z}=\textbf{w}+\beta \Xi \textbf{x}$.
\end{proposition}
\begin{proof}
Now the Hessian is positive-definite iff $\lambda_{\min}(I-\beta\Xi^{T}J_{s}\Xi)>0$. As $I$ is an identity matrix, we can rewrite this condition as $1+\lambda_{\min}(-\beta\Xi^{T}J_{s}\Xi)>0$. From this we see that
\begin{equation} \label{eqn:primary_eigen_condition}
 \lambda_{\max}(\Xi^{T}J_{s}\Xi)<1/\beta,
\end{equation}
where we've employed the scaling of an eigenvalue and the fact that when we take out the minus sign, we change from a maximal to a minimal eigenvalue.

Recall that our matrix $\Xi^{T}J_{s}\Xi$ is real and symmetric, and so we may use the variational characterization of the maximal eigenvalue, namely for a real and symmetric matrix $M$, the maximal eigenvalue is given by the maximum over all vectors $v$ of the quadratic form $(v^{T}Mv)/(v^{T}v)$ As such we may rewrite Eq.~\ref{eqn:primary_eigen_condition} as 
\begin{equation} \label{eqn:bound_on_quad_var}
 \max_{v \in \mathbb{R}^{N}}[\dfrac{v^{T}\Xi^{T}J_{s}\Xi v}{v^{T}v}]<\dfrac{1}{\beta}.
\end{equation}
Now consider an arbitrary vector $v$, and let $z=\Xi v$,
note that $||z|| \leq ||\Xi||_{op}||v||$, where $||.||_{op}$, is the operator norm\footnotemark[\value{footnote}] of $\Xi$. Then we have
\begin{equation}\label{eqn:bound_on_quad_var_2}
 \dfrac{v^{T}\Xi^{T}J_{s}\Xi v}{v^{T}v}=\dfrac{z^{T}J_{s}z}{v^{T}v}=\dfrac{||z||^{2}}{||v||^{2}}\dfrac{z^{T}J_{s}z}{||z||^{2}} \leq ||\Xi||_{op}^{2}\dfrac{z^{T}J_{s}z}{||z||^{2}} \leq ||\Xi||_{op}^{2} \max_{z \in \mathbb{R}^{K}}[\dfrac{z^{T}J_{s}z}{||z||^{2}}]
\end{equation}
Recall that $J_{s}=D_{s}-s \otimes s$ and thus in particular $J_{s}$ is a real symmetric matrix, so that the variational description of the maximal eigenvalue is applicable to it. Thus by combining Eq.~\ref{eqn:bound_on_quad_var}, Eq.~\ref{eqn:bound_on_quad_var_2} and the variational definition of the maximal eigenvalue we have
$$\dfrac{v^{T}\Xi^{T}J_{s}\Xi v}{v^{T}v} \leq ||\Xi||_{op}^{2} \lambda_{\max}(J_{s}).$$
Taking the maximum over $v \in \mathbb{R}^{N}$ we have that
$$\lambda_{\max}(\Xi^{T}J_{s}\Xi)<||\Xi||_{op}^{2} \lambda_{\max}(J_{s}),$$
and as such we arrive at a sufficient condition for $D^{2}V>0$, namely $\lambda_{\max}(J_{s})<||\Xi||^{-2}_{op}\beta^{-1}$.
\end{proof}
\par From this proposition we see that when the maximal eigenvalue of the Jacobian of the softmax, $J_{s}$ is smaller than $||\Xi||^{-2}_{op}\beta^{-1}$ then the Hessian of the energy functional $D^{2}V$ is positive-definite. As such we proceeded to investigate the form of the eigenvalues of $J_{s}$. We do this by breaking it down into several propositions, for simplicity. The result for the maximal eigenvalue, which we are looking for is then presented in Proposition \ref{prop:maximal_eigenvalue_final}.

\begin{proposition}
 $ \forall \textbf{x} \in \mathbb{R}^{N}, 0$ is always an eigenvalue of $J_{s}.$
\end{proposition}
\begin{proof}
 To see this take the vector of $1$'s, in $\mathbb{R}^{K}$, which we denote by $\mathbf{1}$ i.e. $\mathbf{1}=(1,1,...,1)$. Then as $J_{s}=D_{s}-s \otimes s$, the $i^{th}$ component of $J_{s}\mathbf{1}$ is
 $$[J_{s}\mathbf{1}]_{i}=[D_{s}\mathbf{1}]_{i}-[(s\otimes s)\mathbf{1}]_{i}=s_{i}-\sum_{k=1}^{K}s_{i}s_{k}.$$
Now note that as the softmax is a probability vector, we have $\sum_{k}s_{k}=1$ and thus
$$[J_{s}\mathbf{1}]_{i}=s_{i}-s_{i}=0.$$
Hence $J_{s}\mathbf{1}=\mathbf{0}$ and $0$ is an eigenvalue of $J_{s}$. 
\end{proof}
Now we proceed to derive an implicit equation for the eigenvalues in the specific case when the entries of the softmax are all different.
\begin{proposition} \label{prop:non-equal_eigen}
 The non-zero eigenvalues of $J_{s}[w+ \beta \Xi\textbf{x}]$ are given by the solutions to the equation on the softmax entries 
 \begin{equation}\label{eqn:eigen_fraction_cond}
 \sum_{i=1}^{K}\dfrac{s_{i}}{s_{i}-\lambda}=0, 
 \end{equation}
where $s=(s_{1},...,s_{K})$, $s_{i}=[\text{softmax}(\textbf{w}+\beta\Xi \textbf{x})]_{i}$, provided for all $i,j$, such that $ i \neq j$ we have $s_{i} \neq s_{j}$.(I.e. all the entries of the softmax are unique).
\end{proposition}

\begin{proof}
Take $z$ an eigenvector and $\lambda$ an eigenvalue of $J_{s}$. By definition we have that $D_{s}z-(s\otimes s )z=\lambda z$. Writing this in component form we have that,
 $$ s_{i}z_{i}-\sum_{k=1}^{K}s_{i}s_{k}z_{k}=\lambda z_{i} \hspace{0.2cm} \forall i \in \{1,...,K\}.$$
Now set $\bar{z} :=\sum_{k} s_{k}z_{k}$, to get 
\begin{equation}\label{eqn:main_eigen_equation}
(s_{i}-\lambda)z_{i}=s_{i} \bar{z} \hspace{0.2cm} \forall i \in \{1,...,K\}.
\end{equation}

Summing over the $i$'s in Eq.~\ref{eqn:main_eigen_equation}, and recalling that as softmax is a probability vector, $\sum s_{i}=1$, we get 
$$\bar{z}-\lambda\sum z_{i}=\bar{z},$$
and as such either $\lambda=0$, which we already know is an eigenvalue, or 
\begin{equation} \label{eqn:z_zero_sum}
 \sum z_{i}=0.
\end{equation}
As we are in the non-zero eigenvalue case, we know that Eq.~\ref{eqn:z_zero_sum} holds.
Now from Eq.~\ref{eqn:main_eigen_equation} we also have that for a given $i \in \{1,...,K\}$ either $s_{i}= \lambda$ or $z_{i}=s_{i}\bar{z}/(s_{i}-\lambda)$. 
\par Let us examine the first case with the $i$, for which $s_{i}=\lambda$, fixed. Now recall that in this proposition we are considering the case that $s_{i} \neq s_{j}$ for any $i \neq j$. As such for all $j \in \{1,...,K\}\backslash \{i\}$, $s_{j}\neq \lambda$. Now returning back to Eq.~\ref{eqn:main_eigen_equation} for the $i^{\text{th}}$ case and utilizing $s_{i}=\lambda$, we get $\bar{z}=0$. But then for any other $j \neq i$, Eq.~\ref{eqn:main_eigen_equation} gives 
$$(s_{j}-\lambda)z_{j}=0,$$ 
and as $s_{j} \neq \lambda$, $z_{j}=0$. But then from Eq.~\ref{eqn:z_zero_sum}, we must have that $z_{i}=0$ and as such $z=0$ and so cannot be an eigenvector. 
\par As such we conclude that under the conditions of our propositions the second case must hold, and in fact it must hold for every $i \in \{1,...,K\}$ i.e. 

$$z_{i}=\dfrac{s_{i}\bar{z}}{s_{i}-\lambda} \hspace{0.2cm} \forall i \in \{1,...,K\}.$$
Summing this equation over the $i$'s and by Eq.~\ref{eqn:z_zero_sum}, we have that any eigenvalue must satisfy
$$\sum_{i=1}^{K} \dfrac{s_{i}\bar{z}}{s_{i}-\lambda}=0,$$
and thus
$$\sum_{i=1}^{K} \dfrac{s_{i}}{s_{i}-\lambda}=0.$$
\par Now to finish our proof, we only need to check that any solution of Eq.~\ref{eqn:eigen_fraction_cond} is an eigenvalue. Hence suppose $\lambda$ satisfies our equation and take a vector $z$ to be given by $z_{i}=s_{i}/(s_{i}-\lambda)$. Then 
$$[J_{s}z- \lambda z]_{i}=[D_{s}z-(s\otimes s)z-\lambda z]_{i}=\dfrac{s_{i}^{2}}{s_{i}-\lambda}- s_{i}\sum_{k=1}^{K}\dfrac{s_{k}^{2}}{s_{k}-\lambda}-\dfrac{\lambda s_{i}}{s_{i}-\lambda}=s_{i}\left(1-\sum_{k=1}^{K}\dfrac{s_{k}^{2}}{s_{k}-\lambda} \right)$$
Therefore
$$[J_{s}z- \lambda z]_{i}=s_{i}\left(\sum_{k=1}^{K}s_{k}-\sum_{k=1}^{K}\dfrac{s_{k}^{2}}{s_{k}-\lambda} \right)= s_{i}\left(-\sum_{k=1}^{K}\dfrac{\lambda s_{k}}{s_{k}-\lambda} \right)=0$$ 
and so $\lambda$ is indeed an eigenvalue. 
\end{proof}
With this proposition, the only case left unexamined, is the case where some of the softmax entries are equal. This is precisely the subject of the following proposition.
\begin{proposition}
 Suppose we have equal entries in the softmax. Denote them by $\sigma_{1},..,\sigma_{r}$, i.e. for any $m$, $\sigma_{m}$ is such that there exists some $i,j$ ($i \neq j$) for which $s_{i}=s_{j}=\sigma_{m}$. (Note it can be the case that when more than 2 entries are equal, e.g. $s_{1}=s_{2}=s_{3}$, and we will still assign the same sigma to them, i.e. $\sigma_{1}=s_{1}=s_{2}=s_{3}$ and so on for more identical entries.) Then the non-zero eigenvalues are given by the solutions of Eq~\ref{eqn:eigen_fraction_cond} in Proposition \ref{prop:non-equal_eigen} and $\lambda_{1}=\sigma_{1},...,\lambda_{r}=\sigma_{r}$.
\end{proposition}
\begin{proof}
\par We start by showing that any repeated entry is indeed an eigenvalue, $s_{i_{1}}=s_{i_{2}}=...=s_{i_{M}}=:\sigma$, for some collection of indices $\{i_{1},...,i_{M}\}$, which we denote by $\mathcal{I}$. For our proof we will split our indices $\{1,...,K\}$ into those in $\mathcal{I}$ and those not in $\mathcal{I}$.

\par We first look at the case when $\{1,...,K\}\backslash \mathcal{I}= \emptyset$, to illustrate the idea. In this scenario all entries of the softmax are equal and take the value $\sigma$. In this case take our candidate for eigenvector to be, $e_{l} -e_{i}$, where $i \neq l$ and $e_{i}$ is one of the canonical basis vectors, i.e. the value of the $j^{\text{th}}$ component is $[e_{i}]_{j}=\delta_{ij}$. By looking at the $k^{\text{th}}$ component of the vector$[J_{s}(e_{l}-e_{i})]_{k}=[(D_{s}-(s \otimes s))(e_{l}-e_{i})]_{k}$, we have
$$[(D_{s}-(s \otimes s))(e_{l}-e_{i})]_{k}=s_{l}\delta_{lk}-s_{k}\delta_{ik}-\sum_{m=1}^{K}s_{k}s_{m}\delta_{lm} +\sum_{m=1}^{K}s_{k}s_{m}\delta_{im}=\sigma \delta_{lk}-\sigma\delta_{ik}-\sigma^{2}+\sigma^{2}.$$
As such we have
$$[(D_{s}-(s \otimes s))(e_{l}-e_{i})]_{k}=\sigma \delta_{lk}-\sigma\delta_{ik}=\sigma[e_{l}-e_{i}]_{k},$$
and we have that 
$J_{s}(e_{l}-e_{i})=\sigma(e_{l}-e_{i}).$
Thus we have shown that $\sigma$ is an eigenvalue of $J_{s}$ (and in fact we've shown that it is the only non-zero eigenvalue, as $e_{l}-e_{i}$ span a $(K-1)$-dimensional subspace).

\par Now we move onto the general case, suppose $\{1,...,K\}\backslash\mathcal{I} \neq \emptyset$. We denote by $z$, our candidate for our eigenvector and we take it to have the following form. We pick two arbitrary indices $a,b\in \mathcal{I}$. (Recall that by definition of $\mathcal{I}$, we must have $|\mathcal{I}| \geq 2$). Set 
$$z_{a}=1, z_{b}=-1 \text{ and } \forall( r\neq a,b ) \text{ set }z_{r}=0$$
Let us examine the vector $(D_{s}-(s\otimes)s)z$ for our choice of $z$. Consider a component $k$. We have 
$$[(D_{s}-(s\otimes s))z]_{k}=s_{k}z_{k}-\sum_{m=1}^{K}s_{r}s_{m}z_{m}= s_{k}(\delta_{ka}-\delta_{kb})-s_{k}(s_{a}-s_{b}).$$
Proceeding further with the calculation we get
$$[(D_{s}-(s\otimes s))z]_{k}=s_{k}(\delta_{ka}-\delta_{kb})-s_{k}(\sigma-\sigma).$$
where we have utilised the definition of $\sigma$. Now this gives us that for our choice of $z$, we get
$$J_{s}z=[D_{s}-(s\otimes s)]z=\sigma z,$$
where the last equality, crucially, follows from the fact that when $k \neq a,b , (\delta_{ka}=0=\delta_{kb})$ and when $k=a$ or $k=b$, $s_{k}=\sigma$. Thus $z$ is indeed an eigenvector, with the eigenvalue $\sigma$. Note also that this eigenspace will be of dimension at least $|\mathcal{I}|-1$, due to the free choices of $a,b$.
\par Finally, we need to show that the rest of the eigenvalues are given as before by the solutions to Eq.~\ref{eqn:eigen_fraction_cond}. We do this by the same calculation as in Proposition \ref{prop:non-equal_eigen}. Let $\lambda $ be a solution \footnote{Note that the only case when there are no solutions to the Eq.~\ref{eqn:eigen_fraction_cond} is when all entries are equal} of Eq.~\ref{eqn:eigen_fraction_cond}. Take $z$ to be a vector given by $z_{i}=s_{i}/(s_{i}-\lambda)$. By doing the exact same calculation as in Proposition \ref{prop:non-equal_eigen}, we get
$$[J_{s}z- \lambda z]_{i}=[D_{s}z-(s\otimes s)z-\lambda z]_{i}=\dfrac{s_{i}^{2}}{s_{i}-\lambda}- s_{i}\sum_{k=1}^{K}\dfrac{s_{k}^{2}}{s_{k}-\lambda}-\dfrac{\lambda s_{i}}{s_{i}-\lambda}=s_{i}\left(1-\sum_{k=1}^{K}\dfrac{s_{k}^{2}}{s_{k}-\lambda} \right)$$
Therefore
$$[J_{s}z- \lambda z]_{i}=s_{i}\left(\sum_{k=1}^{K}s_{k}-\sum_{k=1}^{K}\dfrac{s_{k}^{2}}{s_{k}-\lambda} \right)= s_{i}\left(-\sum_{k=1}^{K}\dfrac{\lambda s_{k}}{s_{k}-\lambda} \right)=0$$ 
and so $\lambda$ is indeed an eigenvalue. (Note also that $\lambda$ has an eigenspace of dimension at least $1$, and now summing over all the eigenspaces dimensions, we get that we have a complete characterization of the eigenvalues of $J_{s}$).
\end{proof}

As such we can state the following about the maximal eigenvalue:
\begin{proposition}\label{prop:maximal_eigenvalue_final}
 The maximal eigenvalue of $J(s)$ is given by the maximal component of the soft-max, $s_{i}$, if there is another $s_{j}$ such that $s_{j}=s_{i}$ and otherwise, it is given by the maximal solution of the equation
 \begin{equation}\tag{\ref{eqn:eigen_fraction_cond}}
 \sum_{i=1}^{K}\dfrac{s_{i}}{s_{i}-\lambda}=0.
 \end{equation}
 $$$$
\end{proposition}

\subsubsection{Calculation of the estimated separatrix}
We are interested in finding the separating surfaces, which form the boundaries between the basins of attractions and by above we expect that $\lambda_{\max}(J)$ can be used as a good estimate for this. To find the separatrix, we employ the following observation. 

\par Consider the case when we are in a region where a single pattern $\Xi_{m}$ is dominating the dynamics, i.e. $s_{m} \approx 1$. Then we can write for $j \neq m$, $s_{j}=\epsilon\rho_{j}$, where $\epsilon>0$ is small and $\rho_{j}$ is a rescaled variable, so that $\rho_{j}=\mathcal{O}(1)$. This then means $s_{m}=1-\epsilon(\sum_{j \neq m} \rho_{j})$, and so Eq.~\ref{eqn:eigen_fraction_cond} becomes
$$\dfrac{1-\epsilon\sum_{j\neq m}\rho_{j}}{1-\epsilon\sum_{j \neq m}\rho_{j}- \lambda(\epsilon)}+\sum_{j\neq m}\dfrac{\epsilon\rho_{j}}{\epsilon\rho_{j}-\lambda(\epsilon)}=0,$$
where we have written $\lambda(\epsilon)$ to signify that $\lambda$ is dependent on $\epsilon$. Multiplying through to eliminate the fractions we get
$$(1-\epsilon\sum_{j\neq m}\rho_{j})\prod_{j\neq m}(\epsilon\rho_{j}-\lambda(\epsilon))+(1-\epsilon\sum_{j \neq m}\rho_{j}-\lambda(\epsilon))\sum_{j\neq m}\epsilon\rho_{j}\prod_{k \neq j,m}(\epsilon\rho_{m}-\lambda(\epsilon))=0.$$

Employing an asymptotic expansion for $\lambda(\epsilon)$ of the form $\lambda(\epsilon)=\lambda_{0}+\epsilon\lambda_{1}+\mathcal{O}(\epsilon^{2})$, we have
$$(1-\epsilon\sum_{j\neq m}\rho_{j})\prod_{j\neq m}(\epsilon\rho_{j}-\lambda_{0}-\epsilon\lambda_{1})+(1-\epsilon\sum_{j \neq m}\rho_{j}-\lambda_{0}-\epsilon\lambda_{1})\sum_{j\neq m}\epsilon\rho_{j}\prod_{k \neq j,m}(\epsilon\rho_{m}-\lambda_{0}-\epsilon\lambda_{1})+\mathcal{O}(\epsilon^{2})=0.$$
Now focusing on the $\mathcal{O}(1)$ we get,
$$(-\lambda_{0})^{K-1}=0.$$
Therefore $\lambda_{0}=0$. So at a steady state we have $\lambda_{\max} \approx 0$.
\par Now imagine following a path $\pi(r)$ between two points $p$ and $q$, parametrized by $r$ so that $\pi(0)=p$ and $\pi(1)=q$. If $p$ and $q$ are two stable steady states, then considering $\lambda_{max}(r)$ as a function of $r$, we have $\lambda_{max}(0)\approx 0 \approx \lambda_{max}(1)$. As between two stable steady states, there will necessarily be a separating surface, on that surface $D^{2}V$ is not positive definite and hence in particular, $\lambda_{max}>1/\beta$. Moreover this means that $\lambda_{max}(r)$ must have at least one turning point, and we will utilise it as our indicator of the position of the separatrix. As such we are interested in the points where
$$\dfrac{\partial \lambda_{max}(r)}{\partial r}=0$$
by the chain rule we get,

$$\dfrac{\partial \lambda_{max}(r)}{\partial r}=\sum_{i=1}^{K}\dfrac{\partial\lambda_{max}}{\partial s_{i}}\dfrac{\partial s_{i}}{\partial r}.$$
By differentiating Eq.~\ref{eqn:eigen_fraction_cond}, we have that $\lambda_{\max}$ has to satisfy,
$$\dfrac{1}{s_{i}-\lambda_{\max}}-\dfrac{s_{i}}{(s_{i}-\lambda_{\max})^{2}}(1-\dfrac{\partial\lambda_{\max}}{\partial s_{i}})=0,$$
and so 
$\dfrac{\lambda_{\max}}{s_{i}}=\dfrac{\partial\lambda_{\max}}{\partial s_{i}}.$ Plugging this back in we have
$\dfrac{\partial \lambda_{max}(r)}{\partial r}=\lambda_{\max}\sum_{i=1}^{K}\dfrac{1}{s_{i}}\dfrac{\partial s_{i}}{\partial r}$. As such we will have a turning point when $\lambda_{\max}=0$, which is the case when we are at a steady state or 
$$0=\sum_{i=1}^{K}\dfrac{1}{s_{i}}\dfrac{\partial s_{i}}{\partial r}.$$

\par Now recall that our position on the path $\textbf{x}$ is given by $\textbf{x}=\pi(r)$. By differentiating the soft-max function we have

$$\dfrac{\partial s_{i}}{\partial r}= -\dfrac{e^{w_{i}+\sum_{m}\xi_{im}x_{m}}}{Z^{2}}\dfrac{\partial Z}{\partial r}+\dfrac{e^{w_{i}+\sum_{m}\xi_{im}x_{m}}}{Z}\beta\sum_{m}\xi_{im}\dfrac{\partial x_{m}}{\partial r},$$
and thus, by denoting $\pi_{m}(r)$ as the $m^{\text{th}}$ co-ordinate of the path $\pi(r)$, we get
$$\dfrac{\partial s_{i}}{\partial r}= -\dfrac{s_{i}}{Z}\sum_{j}e^{w_{j}+\beta\sum_{m}\xi_{jm}x_{m}}(\beta\sum_{m}\xi_{jm}\dfrac{\partial\pi_{m}}{\partial r})+s_{i}\beta\sum_{m}\xi_{im}\dfrac{\partial \pi_{m}}{\partial r},$$
and therefore we have
$$\dfrac{\partial s_{i}}{\partial r}= -s_{i}\sum_{j}s_{j}(\beta\sum_{m}\xi_{jm}\dfrac{\partial\pi_{m}}{\partial r})+s_{i}\beta\sum_{m}\xi_{im}\dfrac{\partial \pi_{m}}{\partial r},$$

$$\dfrac{\partial s_{i}}{\partial r}=\beta s_{i} \left[ \sum_{m=1}^{N} \xi_{im}\dfrac{\partial \pi_{m}}{\partial r}-\sum_{j=1}^{K}\sum_{m=1}^{N}s_{j}\xi_{jm}\dfrac{\partial \pi_{m}}{\partial r} \right].$$
Rewriting this in terms of patterns $\Xi_{k}$, we get
$$\dfrac{\partial s_{i}}{\partial r}=\beta s_{i} \left[ \Xi_{i}-\sum_{k=1}^{K}s_{k}\Xi_{k} \right] \cdot \dfrac{\partial \pi}{\partial r}.$$
Therefore the derivative of our maximal eigenvalue is given by 
$$\dfrac{\partial \lambda_{max}(r)}{\partial r}=\lambda_{\max}\beta\sum_{i=1}^{K}\left[ (1-Ks_{i})\Xi_{i} \right] \cdot \dfrac{\partial \pi}{\partial r}.$$
And so we have a turning point, and thus we predict a separatrix when
\begin{equation} \label{eqn:separatrices_arbitr_path}
 0=\sum_{i=1}^{K}\left[ (1-Ks_{i})\Xi_{i} \right] \cdot \dfrac{\partial \pi}{\partial r}.
\end{equation}
We note that when $\pi(r)$ is just a straight line between $p$ and $q$, we have
\begin{equation} \label{eqn:separatrices_straight_line}
 0=\sum_{i=1}^{K}\left[ (1-Ks_{i})\Xi_{i} \right] \cdot (p-q).
\end{equation}

\subsection{Separatrix between two patterns}

\par In this subsection we aim to derive the equation approximating the separatrix, between two patterns. Suppose we have two patterns $\Xi_{i},\Xi_{k}$, with the angle between them being $\angle(\Xi_{i},\Xi_{k})=\mu$ (note that this will also be the relevant effective dynamics when the softmax entries of the other patterns are small). Further let us assume that $\beta$ is such that the bifurcation separating $\Xi_{i}$ from $\Xi_{k}$ occurred. As before we assume that the patterns have norm unity.
From Eq.~\ref{eqn:separatrices_straight_line}, we predict that a point on our separatrix must satisfy
$$0=\left[ (1-2s_{i})\Xi_{i}+ (1-2s_{k})\Xi_{k}\right] \cdot (p-q).$$
And so $0=\left[ (s_{k}-s_{i})\Xi_{i}+ (s_{i}-s_{k})\Xi_{k}\right] \cdot (p-q)=(s_{k}-s_{i})(\Xi_{i}-\Xi_{k}) \cdot (p-q)$.
Therefore provided $\Xi_{i}$ and $\Xi_{k}$ are linearly independent,
we have that $s_{i}=s_{k}$ and so our estimate for the separatrix equation is 
$$w_{i}+ \beta \Xi_{i} \cdot \textbf{x} = w_{k} + \beta \Xi_{k} \cdot \textbf{x}$$

\par As such we get that any point on the separatrix must satisfy $\textbf{x} \cdot (\Xi_{i}-\Xi_{k})=-(w_{i}-w_{k})/\beta$, which we recognize as an equation for a hyperplane. 
\par To see how far from the mid-point between two patterns the separatrix, we look at $\textbf{x}=(\Xi_{i}+\Xi_{k})/2+r(\Xi_{i}-\Xi_{k})$, i.e. the straight line connecting the two patterns, we find that:
$$r||(\Xi_{i}-\Xi_{k})||^{2}=-\dfrac{w_{i}-w_{k}}{\beta}$$
and so 
$$r=-\dfrac{w_{i}-w_{k}}{\beta(2-2\cos(\mu))}$$
Therefore the offset from the middle point between two patterns ($(\Xi_{i}+\Xi_{k})/2$) will be given by
$$\Delta=\dfrac{(w_{i}-w_{k})||(\Xi_{i}-\Xi_{k})||}{\beta(2-2\cos(\mu))}.$$
Finally note that when we only have two patterns we may write $\Xi_{i}=(\cos(\mu/2),\sin(\mu/2)), \Xi_{r}=(\cos(\mu/2),-\sin(\mu/2))$ and thus
\begin{equation} \label{eqn:offset_of_separatrix}
 \Delta=\dfrac{(w_{i}-w_{k})\sin(\mu/2)}{\beta(1-\cos(\mu))},
\end{equation}
where the offset moves the separatrix $\Delta$ closer to $\Xi_{k}$. A comparison between our analytical prediction and numerical results, and an illustration of the w-modulated offset can be seen in Figure \ref{fig:w_moves_separatrix}.

\begin{figure}[H]
 \centering
 \includegraphics[width=0.98\linewidth]{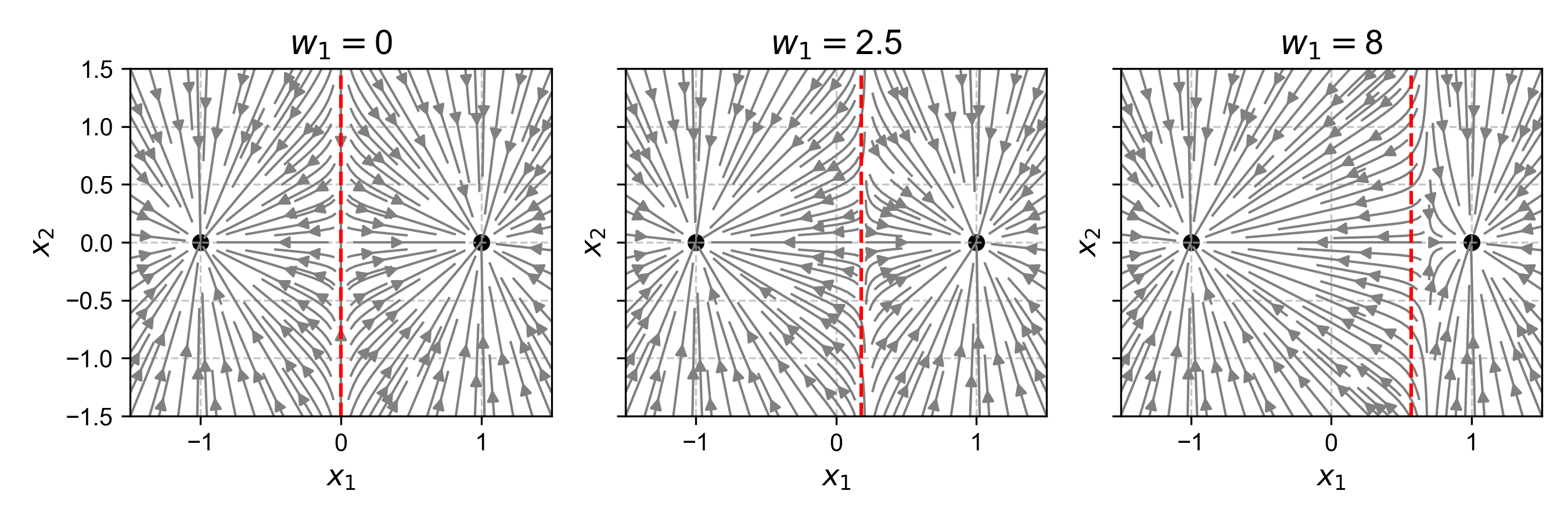}
\caption{\textbf{An example of how modulating $w_{1}$ moves the separatrix between two patterns} In this example we have 2 patterns, represented by black dots. $\Xi_{1}$ is the pattern on the left. The red dashed line represents the analytic estimate for the separatrix. In this simulation $\beta=7$, $\textbf{w}=(w_{1},0)$.}
 \label{fig:w_moves_separatrix}
\end{figure}
\subsubsection{Effect of a far away pattern on the separatrix}
\par In this section, we will show that an addition of a pattern that is far away from the two patterns doesn't affect the position of the separatrix between the two close patterns. In this section we suppose that all patterns are of unit magnitude. Consider then two isolated patterns, $\Xi_{i},\Xi_{k}$, with an angle $\mu$ between them and a pattern far away, $\Xi_{r}$, say. As in section \ref{section:Dynamcis_vicinity}, we suppose that $\mu \ll \angle(\Xi_{i},\Xi_{r}),\angle(\Xi_{k},\Xi_{r}))$. We further suppose that $\beta$ is large enough so that the bifurcation separating $\Xi_{i}$ and $\Xi_{k}$ occurred. 
\par To see this let us consider our separatrix equation given by Eq.~\ref{eqn:separatrices_straight_line}, when $\pi$ is just a straight line between $p=\Xi_{i}$ and $q=\Xi_{k}$, i.e. we are looking at the straight line between the two patterns. We have
$$ 0=\sum_{m=1}^{K}\left[ (1-Ks_{m})\Xi_{m} \right] \cdot (\Xi_{i}-\Xi_{k}).$$
As $K=3$, in our case, expanding we get
$$ 0=\left[(1-3s_{i})\Xi_{i}+(1-3s_{k})\Xi_{k}+(1-3s_{r})\Xi_{r} \right] \cdot (\Xi_{i}-\Xi_{k}).$$
Now as our patterns have a unity magnitude, we get
$$0=(1-3s_{i})(1-\cos(\mu))+(1-3s_{k})(\cos(\mu)-1)+(1-3s_{r})(\cos(\angle(\Xi_{i},\Xi_{r}))-\cos(\angle(\Xi_{k},\Xi_{r}))),$$
and so simplifying we have
\begin{equation} \label{eqn:calc_of_sep_angles_three}
    0=3(s_{k}-s_{i})(1-\cos(\mu))+(1-3s_{r})(\cos(\angle(\Xi_{i},\Xi_{r}))-\cos(\angle(\Xi_{k},\Xi_{r}))).
\end{equation}
Now as $\Xi_{r}$ is the isolated pattern, we will have that $\angle(\Xi_{i},\Xi_{r}) \approx \angle(\Xi_{k},\Xi_{r})$ and thus Eq.~\ref{eqn:calc_of_sep_angles_three} becomes
$$ 0\approx3(s_{k}-s_{i})(1-\cos(\mu)),$$
and so $s_{k}=s_{i}$, and the position of the separatrix is thus the same as in the case where we only considered two patterns.
\subsection{Control of progenitor basins by modulation of $\textbf{w}$}

\par Consider now modulating one component of $\textbf{w}$, $w_{i}$. And suppose we coarse-grain the patterns with the indices given by $\mathcal{L}$ such that, $i \in \mathcal{L}$. Then comparing the separatrix between the coarse-grained pattern, $\Xi^{*}$ and another pattern $\Xi_{r}$, we have from Eq.~\ref{eqn:offset_of_separatrix} that the offset from the separatrix is given by 
$$\Delta=\dfrac{(w^{*}-w_{k})\sin(\mu/2)}{\beta(1-\cos(\mu))}=\dfrac{(\log(\sum_{j \in \mathcal{L}}e^{w_{j}})-w_{k})\sin(\mu/2)}{\beta(1-\cos(\mu))},$$
where $\mu$ is the angle between $\Xi^{*}$ and $\Xi_{k}$. Note that we have the following bounds on $\log(\sum_{j \in \mathcal{L}}e^{w_{j}})$, namely 
$$ \max_{j \in \mathcal{L}}(w_{j})\leq \log(\sum_{j \in \mathcal{L}}e^{w_{j}}) \leq \max_{j \in \mathcal{L}}(w_{j})+\log(|\mathcal{L}|)$$

Now consider increasing $w_{i}$, starting from $0$, we see that while $w_{i}< \max_{j \in \mathcal{L}}(w_{j})$, it makes negligible contribution to the position of the separatrix of the progenitor pattern. However once $w_{i}=\max(w_{j})$, then it is the main contributor to the position of the separatrix and we have that
$$\Delta \geq \dfrac{(w_{i}-w_{k})\sin(\mu/2)}{\beta(1-\cos(\mu))}.$$
This is precisely what we see in our simulations. We present an example of the modulation of the separatrix of the coarse-grained pattern by $\textbf{w}$ in Figure \ref{fig:w_modul_annealing}.
\begin{figure}[H]
 \centering
 \includegraphics[width=0.6\linewidth]{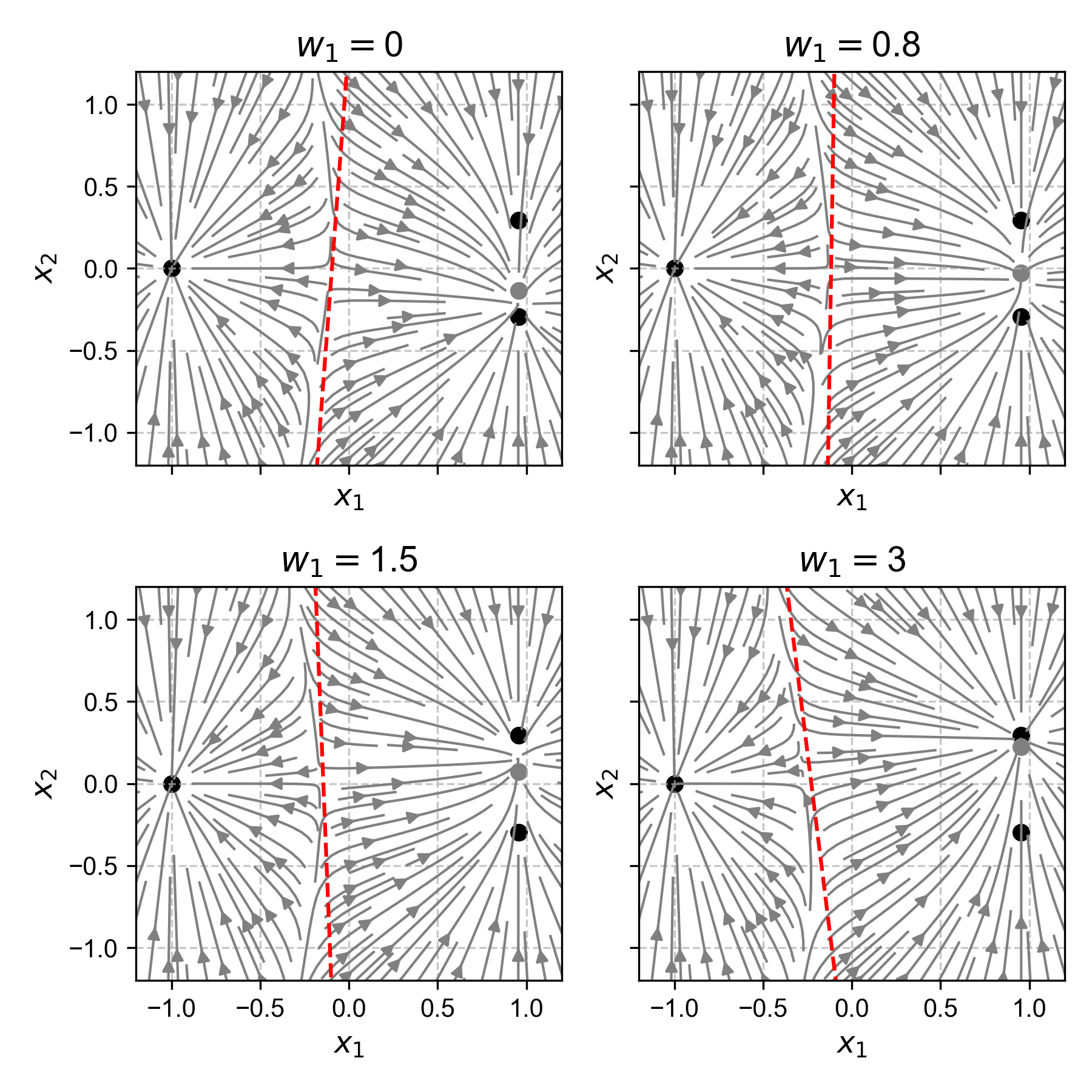}
\caption{\textbf{An example of how modulating $w_{1}$ moves the separatrix of the coarse-grained pattern} In this example we have 3 patterns, represented by black dots, with a coarse-grained pattern (gray dot), stemming from $\Xi_{1}$ and $\Xi_{2}$. $\Xi_{1}$ is the upper right pattern. The red dashed line represents the analytic estimate for the separatrix. In this simulation $\beta=7$, $\textbf{w}=(w_{1},1,0)$. As seen $w_{1}$ has little effect on the position of the separatrix until $w_{1}>w_{2}=1$.}
 \label{fig:w_modul_annealing}
\end{figure}

\section{Local instability and the calculation of $\beta_{\text{crit}}$ for equal-angle patterns}

Consider a set of $K$ patterns $\Xi$ of dimension $N$ for which each pattern is of magnitude unity and at angle $\mu$ relative to all other patterns. Furthermore we assume that $\textbf{w}=0$. As per Section B, since the dynamics are rotationally invariant, we can assume, without loss of generality, that these patterns are given by \emph{any} $K$ dimensional set of vectors for which:
\begin{equation}\label{eqn:equal_pattern_dot_product}
 \Xi_i \cdot \Xi_j = \begin{cases}
 \cos{\mu} & \text{if}\ i\neq j \\
 1 & i=j
 \end{cases}
\end{equation}
as an example, in two dimensions, we have that:
\begin{equation}
 \Xi_1 = (1,0),\; \Xi_2=(\cos{\mu},\sin{\mu})
\end{equation}
while in three dimensions, we have that 
\begin{equation}
 \Xi_1 = (1,0,0),\; \Xi_2=(\cos{\mu},\sin{\mu},0),\; \Xi_3=\left(\cos{\mu},\cos{\mu}\tan{\left(\frac{\mu}{2}\right)},\sqrt{(2 \cos{\mu}+1) \tan ^2\left(\frac{\mu}{2}\right)}\right)
\end{equation}
we further assume that the weights are all identical, and thus the system is entirely symmetric. In this case, our "progenitor" state corresponds to:
\begin{equation}
 \Xi^* = \frac{1}{K}\sum_i \Xi_i
\end{equation}
What is the stability of $\Xi^*$? We can address this by linear stability analysis, by estimating the Jacobian of Eq.~\ref{eq:enhancernet} around $\textbf{x}=\Xi^*$. Since the dynamics are associated with a potential function (symmetric Jacobian) the eigenvalues are real at any evaluated point, and $\Xi^*$ may be a stable fixed point only if all its associated eigenvalues are negative. For two equidistant patterns, these eigenvalues are given by $\lambda_{1,2}=-1,\frac{1}{2} (\beta(1-\cos{\mu } ) -2)$, and there is thus a bifurcation at $\beta(1-\cos{\mu})=2$ (that is, around $\beta \mu^2 = 4)$. For three patterns, the eigenvalues are $\lambda_{1,2,3}=-1,\frac{1}{3} (\beta(1-\cos{\mu } ) -3),\frac{1}{3} (\beta(1-\cos{\mu } ) -3)$, and the bifurcation occurs at around $\beta (1-\cos{\mu}) = 3$. Similarly, for $K$ patterns, we expect the bifurcation to occur at the critical $\beta (1-\cos\mu)=K$, which we indeed verify numerically. Thus, a transition from progenitor to terminal pattern, for an isolated subset of $K$ patterns, occurs at a critical $\beta_{\text{crit}}\approx 2K \mu^{-2}$.

\par In fact we can analytically derive the critical value of $\beta$, $\beta_{\text{crit}}$ for which the bifurcation occurs. The rest of this section is dedicated to showing this result.
\begin{proposition} \label{prop:bifur_equal_angle_patterns}
 For $K$ patterns, with all having an angle $\mu$ between each other, as above, all the eigenvalues of the Jacobian of Eq.~\ref{eq:enhancernet} evaluated at $\textbf{x}=\Xi^{*}$ are negative, provided 
 $$\beta(1-\cos(\mu))<K,$$
and $\cos(\mu) \neq -1/(K-1)$.
\end{proposition}
To prove this we will utilise the following lemma.
\begin{lemma} \label{lemma:lin_indep_pattern}
If a collection of patterns $\Xi_{1},...,\Xi_{K},$ satisfy Eq.~\ref{eqn:equal_pattern_dot_product}, then provided 
$$\cos(\mu) \neq -\dfrac{1}{(K-1)}$$
holds, the collection of patterns is linearly independent.
\end{lemma}
\begin{proof}
Consider the equation $c_{1}\Xi_{1}+...+c_{K}\Xi_{K}=0$, by looking at the inner product with $\Xi_{i}$, we have from Eq.~\ref{eqn:equal_pattern_dot_product}
$$1+\sum_{ j \neq i } c_{j}\cos(\mu)=0.$$
Now let $c=\sum_{j=1}^{K}c_{j},$ then our equation becomes
\begin{equation} \label{eqn:lin_indep_calc}
 c_{i}+(c-c_{i})\cos(\mu)=0.
\end{equation}
Note that this is true for all $i \in \{1,...,K\}$, so summing over the indices we have $c+(Kc-c)\cos(\mu)=0,$ and so
$$c(1+(K-1)\cos(\mu))=0.$$
Thus we must have either $c=0$ or $\cos(\mu)=-1/(K-1)$ and the second case is excluded due to our assumption. So we have $c=0$ and so plugging this into Eq.~\ref{eqn:lin_indep_calc}, we get
$$c_{i}(1-\cos(\mu))=0.$$
As $\cos(\mu)=1$ implies $\mu=0$ (i.e. there is no angle between patterns and so there is only one pattern), and we have $\mu >0$, we conclude that $c_{i}=0$ and so our patterns are linearly independent. 
\end{proof}
We remark that the case that is excluded from the lemma above, namely $\cos(\mu)=-1/(K-1)$ cannot give linear independence. In that case $\Xi^{*}=0$ and so we have a linear dependence of patterns \footnote{However as mentioned previously, as we're in the symmetric case of $\Xi=\mathbf{Q}$, our singular value decomposition would give $0$ as one of the singular values, as we've just seen that in this case the rank of $\Xi$ is $K-1$, and we could further reduce the dimension of our problem by one}. With the lemma in mind we return back to the proof of Proposition \ref{prop:bifur_equal_angle_patterns}.
\begin{proof}[Proof of Proposition \ref{prop:bifur_equal_angle_patterns}]
We are interested in the stability of the progenitor pattern 
$$\Xi^{*}=\dfrac{1}{K}\sum_{i=1}^{K}\Xi_{i}.$$
Recall that we are working in the case, when $\mathbf{Q}=\Xi$, and so we have an energy functional given by Eq.~\ref{eqn:energy_function}. In this case, we know that the Hessian of the energy function (which is the negative of the Jacobian of the dynamics) is given by Eq.~\ref{eqn:hessian_full}, i.e.
$$D^{2}V(\textbf{x})=I-\beta\Xi^{T}J_{s}(\textbf{w}+\beta\Xi \textbf{x})\Xi.$$

\par Now recall that we are assuming $\textbf{w}=\textbf{0}$ and we are interested at $\textbf{x}=\Xi^{*}$. Looking then at the $i^{\text{th}}$ component of the softmax we have
$$[\text{softmax}(\textbf{w}+\beta\Xi \textbf{x})]_{i}=\dfrac{e^{\beta\Xi_{i}\cdot \Xi^{*}}}{Z}=\dfrac{e^{\beta(1+\cos(\mu)(K-1)/K)}}{Z}.$$
In particular this shows that all the components of the softmax are equal at $\textbf{x}=\Xi^{*}$ and thus as it is a probability vector
\begin{equation} \label{eqn:softmax_entries_are_equal}
 [\text{softmax}(\textbf{w}+\beta\Xi \textbf{x})]_{i}=\dfrac{1}{K}.
\end{equation}
Note that the eigenvalues of a real symmetric matrix are all positive if and only if the minimal eigenvalue is positive. Thus, as in Section D, we employ a variational characterisation of the maximal eigenvalue, i.e. the eigenvalues of the Jacobian of our dynamics at $\Xi^{*}$ are all negative (meaning the Hessian of the energy is positive definite) if anf only if Eq.~\ref{eqn:bound_on_quad_var} holds at $\textbf{x}=\Xi^{*}$, namely
$$\max_{v \in \mathbb{R}^{K}}[\dfrac{v^{T}\Xi^{T}J_{s}\Xi v}{v^{T}v}]<\dfrac{1}{\beta},$$
which is equivalent to the condition
\begin{equation} \label{eqn:condition_with_v_unity}
 \max_{v \in \mathbb{R}^{K}, ||v||=1}[v^{T}\Xi^{T}J_{s}\Xi v]<\dfrac{1}{\beta},
\end{equation}
so now throughout we will assume $||v||=1$. If we set $u=\Xi v,$ we have 
$$v^{T}\Xi^{T}J_{s}\Xi v=u^{T}J_{s}u=u^{T}(D_{s}-s\otimes s)u=\sum_{r=1}^{K}u_{r}^{2}s_{r}-\left(\sum_{r=1}^{K}u_{r}s_{r}\right)^{2},$$
where $s_{i}$ is the $i^{\text{th}}$ component of the softmax at $\textbf{x}=\Xi^{*}$. Now by Eq.~\ref{eqn:softmax_entries_are_equal} we get
\begin{equation} \label{eqn:variance_of_u}
 v^{T}\Xi^{T}J_{s}\Xi v=\dfrac{1}{K}\sum_{r=1}^{K}u_{r}^{2}-\dfrac{1}{K^{2}}\left(\sum_{r=1}^{K}u_{r}\right)^{2}.
\end{equation}
$$$$
Now utilising the assumption that $\cos(\mu)\neq -1/(K-1)$, we have by Lemma \ref{lemma:lin_indep_pattern} that $\Xi_{1},...,\Xi_{K}$ are linear independent, and thus form a basis (as we have $K$ dimensional vector space) so we may write
$$v=\sum_{m=1}^{K}\alpha_{m}\Xi_{m},$$
and so 
$$u_{i}=[\Xi v]_{i}=\Xi_{i} \cdot v = \sum_{m=1}^{K} \alpha_{m} \Xi_{i} \cdot \Xi_{m}=\alpha_{i}+\cos(\mu)\sum_{m\neq i} \alpha_{m}=\alpha_{i}+\cos(\mu)(A-\alpha_{i}),$$
where $A=\sum_{m=1}^{K}\alpha_{i}$. Plugging this back into Eq.~\ref{eqn:variance_of_u}, we have
$$v^{T}\Xi^{T}J_{s}\Xi v=\dfrac{1}{K}\sum_{r=1}^{K}(\alpha_{r}+\cos(\mu)(A-\alpha_{r}))^{2}-\dfrac{1}{K^{2}}\left(\sum_{r=1}^{K}(\alpha_{r}+\cos(\mu)(A-\alpha_{r}))\right)^{2}.$$
So setting $S=\sum_{m=1}^{K}\alpha_{m}^{2}$, and expanding we have
$$v^{T}\Xi^{T}J_{s}\Xi v=\dfrac{1}{K}\sum_{r=1}^{K}\left[\alpha_{r}^{2}+2\alpha_{r}\cos(\mu)(A-\alpha_{r}) +\cos^{2}(\mu)(A-\alpha_{r})^{2}\right]-\dfrac{1}{K^{2}}\left[A+\cos(\mu)(AK-A)\right]^{2},$$
and thus
$$ v^{T}\Xi^{T}J_{s}\Xi v=\dfrac{1}{K}[S+2\cos(\mu)(A^{2}-S)+\cos^{2}(\mu)(A^{2}K-2A^{2}+S)]-\dfrac{1}{K^{2}}\left[A+\cos(\mu)A(K-1)\right]^{2}.$$
Expanding further and simplifying we get
$$ v^{T}\Xi^{T}J_{s}\Xi v=\dfrac{1}{K}[S(1-2\cos(\mu)+\cos(\mu)^{2})+A^{2}K\cos^{2}(\mu)]-\dfrac{A^{2}}{K^{2}}\left[1+(K-1)\cos(\mu)\right]^{2},$$
giving
\begin{equation} \label{eqn: calculation of variational eigenvalue both A and S}
 v^{T}\Xi^{T}J_{s}\Xi v=\dfrac{1}{K}[S(1-\cos(\mu))^{2}+A^{2}K\cos^{2}(\mu)]-\dfrac{A^{2}}{K^{2}}\left[1+(K-1)\cos(\mu)\right]^{2}.
\end{equation}
Before proceeding further with this calculation, recall that we are interested in the vectors such that $||v||=1$ and thus
$$1=||v||^{2}=v \cdot v=\left( \sum_{l=1}^{K} \alpha_{l} \Xi_{l}\right)\cdot \left(\sum_{m=1}^{K}\alpha_{m}\Xi_{m}\right)=\sum_{l=1}^{K}\sum_{m=1}^{K} \alpha_{l}\alpha_{m} \Xi_{l} \cdot \Xi_{m}.$$
Thus
$$1=\sum_{l=1}^{K}\left(\alpha_{l}^{2}+\alpha_{l}\sum_{m\neq l}^{K}\alpha_{m} \cos(\mu) \right)=S+\sum_{l=1}^{K}\alpha_{l}\cos(\mu)(A-\alpha_{l})=S+\cos(\mu)(A^{2}-S).$$
And as such we have
$$S=\dfrac{1-A^{2}\cos(\mu)}{1-\cos(\mu)}.$$
Plugging this back into Eq.~\ref{eqn: calculation of variational eigenvalue both A and S},
$$v^{T}\Xi^{T}J_{s}\Xi v=\dfrac{1}{K}[(1-A^{2}\cos(\mu))(1-\cos(\mu))+A^{2}K\cos^{2}(\mu)]-\dfrac{A^{2}}{K^{2}}\left[1+(K-1)\cos(\mu)\right]^{2}.$$
Simplifying this we get
\begin{align*}
v^{T}\Xi^{T}J_{s}\Xi v &= \dfrac{1}{K}[(1-A^{2}\cos(\mu))(1-\cos(\mu))]+A^{2}\cos^{2}(\mu)-\dfrac{A^{2}}{K^{2}}\left[1+(K-1)\cos(\mu)\right]^{2}, \\
&=\dfrac{1}{K}[(1-A^{2}\cos(\mu))(1-\cos(\mu))]+\dfrac{A^{2}}{K^{2}}\left[K^{2}\cos(\mu)^{2}-1+2(K-1)\cos(\mu)-(K-1)^{2}\cos(\mu)^{2}\right], \\
&=\dfrac{(1-\cos(\mu))}{K}-\dfrac{(A^{2}\cos(\mu))(1-\cos(\mu))}{K}+\dfrac{A^{2}}{K^{2}}\left[-1+2(K-1)\cos(\mu)+(2K-1)\cos(\mu)^{2}\right], \\
&=\dfrac{(1-\cos(\mu))}{K}+\dfrac{A^{2}}{K^{2}}\left[-1+(K-2)\cos(\mu)+(3K-1)\cos(\mu)^{2}\right].
\end{align*}
Recall that by Eq.~\ref{eqn:condition_with_v_unity} we want to maximise $v^{T}\Xi^{T}J_{s}\Xi v$ and as such, since $A$ is the only variable at play, this is achieved when $\partial (v^{T}\Xi^{T}J_{s}\Xi v)/\partial A =0,$ i.e. precisely when $A=0$, and thus condition(\ref{eqn:condition_with_v_unity}) now gives
$$\dfrac{1-\cos(\mu)}{K}<\dfrac{1}{\beta}.$$
Therefore whenever this condition holds, the maximal eigenvalue of the Jacobian of our dynamics is negative, and hence all eigenvalues are negative. 
\end{proof}
From this proposition we deuce that when $K=\beta(1-\cos(\mu))$, one of the eigenvalues of the Jacobian of the dynamics, becomes zero. Hence we get a change in the dynamical behaviour and thus a bifurcation.

\section{Simulation}

\subsection{Simulation of annealing in blood production}

Data preparation of gene expression profiles of haematopoietic lineages was based on mouse RNA sequencing data from Haemopedia \cite{choi2019haemopedia} and preprocessed as in Karin \cite{karin2024enhancernet}. Annealing was executed as follows. The dynamics were given by the noisy version of Eq.~\ref{eq:enhancernet}:

\begin{equation}
 d\textbf{x} = \left(\Xi^T \text{softmax}\left(\beta \Xi \textbf{x} + \textbf{w} \right) - \textbf{x} \right)dt + \sigma dW,
 \label{eq:enhancernetnoise}
\end{equation}
where $W$ is a Wiener process and $\sigma$ corresponds to the noise magnitude (set to $\sigma=0.05$). Simulation was performed using the Euler-Maruyama procedure. Annealing was performed with initial conditions as $\textbf{x}(0)=\text{softmax}(\textbf{w}) + \zeta$ where $\zeta$ is a random vector of length $N$ with entries drawn from a uniform distribution over $[0,0.1]$. Annealing itself was executed by a gradual increase of $\beta$ from $\beta=0$ to $\beta=50$ over $t=20$ time units. 

To estimate $\textbf{w}$ for balanced differentiation, we used proportional feedback, namely, after each annealing iteration (performed over $n=1000$ cellular instances), we calculated the error vector:
\begin{equation}
 \text{err} = \textbf{1} - K \rho
 \label{eq:eqerr}
\end{equation}
where $\rho$ is a $K$ dimensional vector (with an entry for each cell type) with its $i$-th entry corresponding to the relative prevalence of cell type $i$ in the annealing output. Note that a perfectly balanced annealing process would produce $\text{err}=\textbf{0}$; and that the control process can be adjusted to produce any output range by replacing $\textbf{1}$ in Eq.~\ref{eq:eqerr} by the desired output profile. The error was then fed to $\textbf{w}$ by the feedback procedure:
$\textbf{w}(t+1)\leftarrow \textbf{w}(t) + G_p \text{err}$
with $G_p$ denoting the feedback gain. In our simulations, the procedure was executed by running $100$ iterations at gain $G_p=0.1$, followed by 100 iterations at $G_p=0.01$ and 100 iterations at $G_p=0.001$.

\subsection{Simulation code}
The code for the simulation, as well as for generating all other figures in the manuscript, is available at \url{www.github.com/karin-lab/hierarchical_control_state_transitions}.

\newpage